\documentclass[onecolumn,prx,longbibliography,nofootinbib, superscriptaddress]{revtex4-2}

\usepackage[dvips]{graphicx} 
\usepackage{amsfonts}
\usepackage{amssymb}
\usepackage{amscd}
\usepackage{amsmath}    
\usepackage{amsthm}
\usepackage{appendix}
\usepackage{bbm}
\usepackage{dsfont}
\usepackage{enumerate}
\usepackage{enumitem}
\usepackage{epsfig}
\usepackage{float}
\usepackage[colorlinks, linkcolor=red, anchorcolor=blue, citecolor=green]{hyperref}
\usepackage{makecell}
\usepackage{subfigure}
\usepackage{subfloat}
\usepackage{xcolor}		
\usepackage{physics}
\usepackage[most]{tcolorbox}
\usepackage{tabularx}
\usepackage{booktabs}

\usepackage{MnSymbol}
\usepackage{tikz}
\usetikzlibrary{quantikz}
\usepackage{tikzpeople}
\usepackage{pgfplots}
\usetikzlibrary{arrows}
\usetikzlibrary{shapes,fadings,snakes}
\usetikzlibrary{decorations.pathmorphing,patterns}
\usetikzlibrary{calc}
\usetikzlibrary{positioning}
\definecolor{amber}{rgb}{1.0, 0.75, 0.0}

\graphicspath{{./figure/}}

\newtheorem{theorem}{Theorem}

\newtheorem{lemma}{Lemma}
\newtheorem{corollary}{Corollary}

\newtheorem{definition}{Definition}
\newtheorem{proposition}{Proposition}

\newtheoremstyle{recap}
  {}                
  {}                
  {\itshape}        
  {}                
  {\bfseries}       
  {.}               
  {.5em}            
  {#3}              

\theoremstyle{recap}
\newtheorem*{thmrecap}{Theorem}

\newcommand{\id}{\mathbb{I}}

\newtcolorbox[auto counter]{mybox}[2][]{
	enhanced,
	breakable,
	colback=blue!5!white,
	colframe=blue!75!black,
	fonttitle=\bfseries,
	title=Box \thetcbcounter: #2,#1
}

\begin{document}
\title{Analog entanglement distillation through Hamiltonian scrambling}

\author{Zitai Xu}
\email{zt3927@umd.edu}
\affiliation{Joint Center for Quantum Information and Computer Science, University of Maryland, College Park, Maryland 20742, USA}

\author{Guoding Liu}
\email{lgd22@mails.tsinghua.edu.cn}
\affiliation{Center for Quantum Information, Institute for Interdisciplinary Information Sciences, Tsinghua University, Beijing, 100084 China}

\author{Xingjian Zhang}
\email{xingjian.zhang@uts.edu.au}
\affiliation{Centre for Quantum Software and Information, University of Technology Sydney, Sydney, Australia}

\begin{abstract}
    Efficient entanglement distillation is a central primitive for quantum information processing and future quantum networks. Conventional protocols rely on digital quantum circuits, requiring compiled multi-qubit gate sequences and fine-grained control that place heavy demands on near-term hardware. Crucially, leading network platforms, such as trapped ions and neutral atoms, naturally support coherent many-body interactions suited for continuous-time analog evolution. Motivated by this, we introduce the Hamiltonian entanglement distillation protocol, which uses information scrambling under native dynamics for analog quantum error detection. We establish a quantitative link between out-of-time-order correlators and distillation performance, showing that efficient scrambling directly yields high output fidelity. This capability is generic across interacting many-body Hamiltonians, with typical random Hamiltonians achieving near-optimal error detection in a constant evolution time independent of system size. Extending this framework to entanglement-based quantum key distribution (QKD), we show that analog error detection combined with classical post-processing saturates the fundamental $33.3\%$ depolarizing noise-tolerance bound. Numerical simulations of device-native Rydberg-atom and trapped-ion Hamiltonians confirm that high-fidelity distillation is achievable within current experimental coherence times. Together, these results establish native analog dynamics as a practical, low-overhead route to entanglement distillation on near-term quantum platforms.
\end{abstract}

\maketitle

\tableofcontents

\section{Introduction}

Quantum information processing promises a new era of computation and communication, with potential application like secure communication~\cite{bennett1984bb84, bennett1993teleportation}, distributed computing~\cite{cirac1999distributed, gottesman1999demonstrating, vanmeter2016path}, quantum simulation and enhanced sensing~\cite{gottesman2012telescopes}. At the heart of many of these applications lies the prerequisite of shared entanglement, the nonlocal resource that makes quantum protocols fundamentally different from classical counterparts.

Despite its central role, entanglement is inherently fragile and susceptible to errors arising from imperfect operations and environmental noise, especially when we want to distribute entanglement between two remote parties. Consequently, raw shared entanglement often suffers from low fidelity and cannot be used directly. To address this, one must employ entanglement distillation, a process that extracts high-fidelity pairs from noisy ensembles. Realizing efficient protocols for this task is crucial yet challenging, requiring joint efforts in both theoretical protocol design and experimental implementation.

On the theoretical front, the field has seen extensive development. Since the seminal work establishing the link between distillation and quantum error correction~\cite{BDSW}, a wide variety of protocols have been proposed~\cite{gottesman2003two-way,devetak2005distillation, chau2002practical, kraus2005lower, watanabe2007key, Du2025AKD, gu2025}. Many of these protocols use stabilizer measurements to detect or correct errors in shared entanglement. Experimental distillation has also been demonstrated across several platforms~\cite{Pan2001, Kalb2017, Kru23, Liu2022, Ecker21single-distillation}. In parallel, quantum error correction has advanced rapidly, with demonstrations of below-threshold logical memories~\cite{Google25Willow} and repeated error correction combined with logical operations in neutral-atom arrays~\cite{Bluvstein2026FaultTolerant}. Achieving this level of error suppression requires extensive experimental control, from precise pulse shaping and gate calibration to circuit compilation and coordinated syndrome measurements. For large quantum networks, requiring sophisticated control at every distillation node compounds the hardware, calibration, and operational demands, raising the resource requirements for building and operating a scalable network. This motivates the search for simpler distillation schemes that retain strong performance while reducing the control resources required at each node.

We pursue this direction using Hamiltonian dynamics in analog quantum systems. Trapped ions and neutral atoms, promising platforms for quantum networks, naturally support coherent many-body interactions~\cite{Bloch2012, Hu2019, Bernien2017}. Their native Hamiltonians can generate collective evolution directly, offering a route to reducing the need for compiled sequences of individually controlled gates~\cite{SimuQ}. In particular, scrambling under these dynamics spreads local errors across many qubits, making them more readily detectable through local measurements. Analog dynamics have already found applications in quantum learning tasks such as shadow tomography~\cite{Tran2023, HamiltonianShadow}. While earlier studies have explored purification using specific native spin-chain or exchange interactions~\cite{Maruyama2008,Auer2014} and scrambling-based error detection in discrete circuit models~\cite{gu2025}, how to design a general distillation protocol for continuous-time many-body Hamiltonians and quantitatively link their scrambling power to distillation performance remains an open question.

In this work, we introduce a Hamiltonian entanglement distillation protocol and present a comprehensive analysis spanning from the theoretical foundations of scrambling-based error detection to practical implementations on near-term analog hardware. Theoretically, we establish a quantitative relationship between out-of-time-order correlators (OTOCs), a standard diagnostic of information scrambling, and the output fidelity and yield of the distillation protocol, proving that strong operator scrambling directly implies high output fidelity. Because generic many-body Hamiltonians are efficient scramblers, this property is ubiquitous: almost all Hamiltonians in the Hilbert space exhibit sufficient scrambling power for high-fidelity distillation. For Hamiltonians drawn from the Gaussian unitary ensemble (GUE) at fixed spectral width, we further prove that near-optimal error detection is achieved at a single $O(1)$ evolution time independent of the number of qubits. Beyond fidelity improvement, we also evaluate our protocol in entanglement-based quantum key distribution (QKD), where noise tolerance in high-noise regimes is a key benchmark. Because one-way error correction in QKD reduces to classical post-processing~\cite{ShorPreskill2000, Liu2022classically}, pairing it with analog error detection under even a simple diagonal Hamiltonian saturates the fundamental $33.3\%$ depolarizing noise-tolerance bound. On the practical side, numerical simulations of device-native Rydberg-atom and trapped-ion Hamiltonians confirm that strong fidelity improvements are achieved within experimentally accessible qubit numbers and coherence times, even without Pauli twirling under non-Pauli amplitude-damping noise. Finally, we evaluate our protocol in practical QKD and quantum repeater settings, demonstrating substantially extended maximum transmission distances.

The remainder of this paper is organized as follows. In Section~\ref{sec:protocol}, we introduce the Hamiltonian entanglement distillation framework and explain the physical intuition behind leveraging analog dynamics. In Section~\ref{sec:main}, we present our core theoretical results on analog error detection, establishing the quantitative connection to OTOCs and proving long-time and constant-time distillation bounds for GUE Hamiltonians. In Section~\ref{sec:tolerance}, we extend the protocol to entanglement-based QKD and noise tolerance, showing that analog error detection under a simple diagonal Hamiltonian combined with classical post-processing saturates the fundamental depolarizing noise threshold. Section~\ref{sec:simulation} presents numerical simulations of device-native Hamiltonians and practical network applications. Finally, Section~\ref{sec:outlook} concludes with a summary and directions for future research.

\section{Hamiltonian-based distillation protocol}\label{sec:protocol}
The entanglement distillation task for two remote parties Alice and Bob, who share $n$ noisy EPR pairs $\tilde{\rho}_{AB}$, is to do some local operations and classical communications to obtain fewer pairs of states with higher fidelity. To see how Hamiltonian dynamics could be leveraged in this process, we first define a Hamiltonian twirling channel $\mathcal{N}_H$ as the ensemble average of the unitary evolution generated by $H$ over a random interaction time: 
\begin{definition}
Given a Hamiltonian $H$ on Hilbert space $\mathcal{H}$ and a time ensemble $\mu$, we define the Hamiltonian twirling channel $\mathcal{N}_{H}$ of a bipartite system  $\mathcal{H}_{AB}=\mathcal{H}^{\otimes 2}$ as,
\begin{equation}\label{eq:twirling}
\begin{split}
    \mathcal{N}_{H}(\rho_{AB})=\int_{t}(e^{-iHt}\otimes e^{iH^*t})\rho_{AB} (e^{-iHt}\otimes e^{iH^*t})^{\dagger}\mu(t).
\end{split}
\end{equation}
\end{definition}
By construction, $\mathcal{N}_{H}$ preserves the maximally entangled state. Specifically, for an $n$-qubit Hamiltonian $H$, one has $\mathcal{N}_{H}(\ketbra{\Phi^+}{\Phi^+}^{\otimes n})=\ketbra{\Phi^+}{\Phi^+}^{\otimes n}$. This follows from the fact that, for any $n$-qubit unitary $U$, the maximally entangled state $\ket{\Phi^+}_{AB}^{\otimes n}$ is invariant under the bilateral action $U \otimes U^*$. That is, $(U_A \otimes U_B^*)\ket{\Phi^+}^{\otimes n}_{AB} = \ket{\Phi^+}^{\otimes n}_{AB}$, where the subscript of $U$ denotes which party it is acting on and $U^*$ denotes the complex conjugate of $U$.

Intuitively, evolving a fixed Hamiltonian for a random time scrambles local errors that do not commute with the Hamiltonian, effectively transforming them into global noise while leaving the EPR state invariant. This approach is reminiscent of the adding noise strategy in two-way entanglement distillation~\cite{Renner2005AddingNoise}. Although the average bit error rate and phase error rate may increase after applying the Hamiltonian twirling $\mathcal{N}_H$, the correlations between the noisy state and the environment are further decoupled. Consequently, when $\mathcal{N}_{H}$ is augmented with an error-detection scheme~\cite{gu2025}, the noisy components can be identified and discarded with high efficiency, yielding EPR pairs with higher fidelity. Figure~\ref{fig:hedp-overview} summarizes the resulting Hamiltonian entanglement distillation protocol, with the detailed steps given in Box~\ref{box:HEDP}.
\begin{figure}[htbp]
    \centering
    \resizebox{0.9\linewidth}{!}{
\begin{tikzpicture}[
    x=1cm,y=1cm,font=\small,
    gate/.style={draw=blue!40!black,fill=blue!4,rounded corners=2pt,
        minimum height=6.5mm,minimum width=9mm,align=center},
    evolution/.style={draw=teal!55!black,rounded corners=3pt,
        left color=blue!18!white,right color=teal!18!white,
        shading angle=0,line width=.8pt,text=blue!35!black,
        minimum height=16mm,minimum width=23mm,align=center},
    measure/.style={draw=orange!65!black,fill=orange!10,
        rounded corners=2pt,minimum width=10mm,minimum height=6mm},
    wire/.style={draw=black!65,line width=.65pt},
    classical/.style={draw=black!60,double=white,double distance=1.2pt,line width=.4pt},
    entanglement/.style={line width=1.1pt,line cap=round},
    heading/.style={font=\scriptsize},
    qubit/.style={circle,draw=blue!55!black,fill=blue!12,
        line width=.7pt,inner sep=2.6pt}
]
    \path[draw=blue!18!gray,fill=blue!5,rounded corners=5pt,line width=.6pt]
        (-.3,-2.6) rectangle (14.7,4.4);
    \path[draw=blue!22!gray,fill=blue!2,rounded corners=5pt,line width=.6pt]
        (-.7,-3) rectangle (14.3,4);
    \path[draw=blue!28!gray,fill=white,rounded corners=5pt,line width=.7pt]
        (-1.1,-3.4) rectangle (13.9,3.6);

    \foreach \x/\y in {13.7/0,14.5/.8} {
        \begin{scope}[shift={(\x,\y)}]
            \path[fill=green!9,draw=green!40!black,line width=.5pt]
                (0,0) circle[radius=.15];
            \draw[draw=green!40!black,line width=1pt,line cap=round,line join=round]
                (-.09,0) -- (-.03,-.06) -- (.09,.07);
        \end{scope}
    }
    \begin{scope}[shift={(14.1,.4)}]
        \path[fill=red!6,draw=red!65!black,line width=.5pt]
            (0,0) circle[radius=.15];
        \draw[draw=red!65!black,line width=1pt,line cap=round]
            (-.07,-.07) -- (.07,.07)
            (-.07,.07) -- (.07,-.07);
    \end{scope}

    \node[heading] at (.4,3.05) {Input};
    \node[heading] at (3,3.05) {Pauli twirl};
    \node[heading] at (5.6,3.05) {Evolution};
    \node[heading] at (8.5,3.05) {Measurement};
    \node[heading] at (12.9,3.05) {Output};
    \node[font=\footnotesize] at (1.8,2.6) {Alice};
    \node[font=\footnotesize] at (1.8,-2.6) {Bob};
    \node[font=\footnotesize] at (.35,-2.95) {$n$ pairs};
    \node[font=\footnotesize] at (12.9,-2.95) {$n-m$ pairs};

    \foreach \x/\y/\bend in {0/2.2/-.7,.65/1.2/.2} {
        \draw[entanglement,draw=blue!45!black]
            (\x,\y) .. controls (\bend,\y) and (\bend,-\y) .. (\x,-\y);
        \foreach \sign in {1,-1} {
            \draw[wire] (\x,\sign*\y) -- (7,\sign*\y);
            \node[qubit] at (\x,\sign*\y) {};
        }
    }

    \foreach \x/\y in {-.5/.55,.32/-.4} {
        \begin{scope}[shift={(\x,\y)}]
            \path[draw=white,line width=2.5pt,fill=yellow!85!orange,
                line join=round]
                (.04,.26) -- (-.15,-.03) -- (-.015,-.015) --
                (-.07,-.27) -- (.16,.075) -- (.025,.05) -- cycle;
            \path[draw=yellow!45!black,line width=.35pt,fill=yellow!85!orange,
                line join=round]
                (.04,.26) -- (-.15,-.03) -- (-.015,-.015) --
                (-.07,-.27) -- (.16,.075) -- (.025,.05) -- cycle;
        \end{scope}
    }

    \node[gate] at (3,2.2) {$P$};
    \node[gate] at (3,1.2) {$Q$};
    \node[gate] at (3,-2.2) {$P^*$};
    \node[gate] at (3,-1.2) {$Q^*$};
    \node[evolution] at (5.6,1.7) {$e^{-iHt}$};
    \node[evolution] at (5.6,-1.7) {$e^{iH^*t}$};
    \node[font=\footnotesize,text=blue!55!black] at (5.6,0) {$t\sim\mu$};

    \foreach \sign in {1,-1} {
        \draw[wire] (7,\sign*2.2) -- (12.9,\sign*2.2);
    }
    \draw[entanglement,draw=green!40!black]
        (12.9,2.2) .. controls (13.6,2.2) and (13.6,-2.2) .. (12.9,-2.2);
    \foreach \sign in {1,-1} {
        \node[qubit,draw=green!40!black,fill=green!12] at (12.9,\sign*2.2) {};
    }

    \foreach \side/\sign in {A/1,B/-1} {
        \node[measure] (meter\side) at (8.5,\sign*1.2) {};
        \draw[wire] (7,\sign*1.2) --
            node[above,font=\footnotesize] {$m$} (meter\side.west);
        \draw[draw=orange!65!black,line width=.6pt]
            ([xshift=-3mm,yshift=-1.2mm]meter\side.center)
            arc[start angle=180,end angle=0,radius=3mm];
        \draw[->,draw=orange!65!black,line width=.6pt]
            ([yshift=-1.2mm]meter\side.center) --
            ([xshift=1.8mm,yshift=2mm]meter\side.center);
    }
    \node[align=center,font=\footnotesize,inner sep=5pt] (compare) at (10.7,0)
        {$a=b$: keep\\[3pt]$a\neq b$: discard};
    \draw[classical,->,rounded corners=2pt] (meterA.east) -|
        node[pos=.25,above,font=\footnotesize] {$a$} (compare.north);
    \draw[classical,->,rounded corners=2pt] (meterB.east) -|
        node[pos=.25,below,font=\footnotesize] {$b$} (compare.south);
    \draw[classical,->] (compare.east) -- (12.95,0);
\end{tikzpicture}}
    \caption{Hamiltonian entanglement distillation. Each card represents a group of $n$ noisy EPR pairs; curved links denote entanglement and lightning marks indicate noise. Alice and Bob apply conjugate Pauli gates and bilateral Hamiltonian evolution, then measure $m$ pairs and compare outcomes via classical communication (double lines). Green checks indicate matching outcomes and retention of the remaining $n-m$ pairs; red crosses indicate rejection of the entire group. The outer and inner wires represent retained and measured qubits, respectively.}
    \label{fig:hedp-overview}
\end{figure}
\begin{figure}[htbp]
\begin{mybox}[label={box:HEDP}]{Hamiltonian entanglement distillation protocol}
\begin{enumerate}[label=(\arabic*)]
    \item {Grouping:}  
    Alice and Bob select an $n$-qubit Hamiltonian $H$ compatible with their devices. 
    Suppose they share $k n$ noisy EPR pairs. They partition these pairs into $k$ groups 
    $\rho_1, \ldots, \rho_k$, each consisting of $n$ EPR pairs.
    \item {Pauli twirling:}  
    For each group, Alice and Bob jointly sample a Pauli operator $P$ from the $n$-qubit Pauli group and apply $P \otimes P^{*}$ to their respective subsystems, obtaining $\rho'_i = (P \otimes P^{*}) \, \rho_i \, (P \otimes P^{*})$.
    \item Hamiltonian twirling: For each group, Alice and Bob jointly sample a time $t$ from time ensemble $\mu$ and simultaneously evolve under $H$ and $-H^*$ on their respective subsystems for time $t$, producing $\rho''_i = (e^{-iHt} \otimes e^{iH^*t})\rho'_i(e^{-iHt} \otimes e^{iH^*t})^{\dagger}$.
    \item {Error detection:}  
    Alice and Bob agree on a POVM $M$ acting on $m$ qubit pairs. For each group, they randomly select $m$ pairs and apply the POVM. The group is kept if the two parties obtain identical measurement outcomes. otherwise, it is discarded.
    \item {Output:}  
    The surviving qubits form EPR pairs with enhanced fidelity, which can be used directly 
    or passed to a subsequent one-way hashing error-correction protocol.
\end{enumerate}
\end{mybox}
\end{figure}

Several notes regarding the scheme are as follows. First, the reason why Alice and Bob perform a random Pauli and its conjugate to their system is basically for analysis simplicity because it turns the noise channel for EPR states into Pauli channel~\cite{BDSW}. This is the most basic and cheap operations and even manageable in many analog simulators. And if, in reality, some analog simulators are not able to perform any Pauli gates, we can omit this step and the protocol should still give good performance as shown in simulation in Sec~\ref{subsec:non-Pauli}. Second, later on in this paper, for theoretical simplicity, we will assume that the Hamiltonian satisfies two properties. First, it is non-degenerate. That is, all its eigenvalues are distinct. Second, the eigenvalues satisfy the non-degenerate gap condition, i.e. $\lambda_i-\lambda_j=\lambda_k-\lambda_l$ with $i\neq j$ and $k\neq l$ implies $(i,j)=(k,l)$. These assumptions are adopted for analytical simplicity. In practice, Hamiltonians with spectral degeneracies, such as the transverse-field Ising Hamiltonian, can also exhibit good performance, as demonstrated numerically in Sec.~\ref{sec:simulation}. We also assume the time measure $\mu$ to be the Fourier transform of delta function $\delta$, which gives the property that $\int_{t}e^{ixt}\mu(t)=1_{x=0}$. In practice this can be approximated by randomly selecting $t$ in an experimental friendly time range $[0,T]$. Third, the protocol requires the ability to implement a given Hamiltonian and its time-reversed evolution. This capability has been demonstrated in analog quantum simulators such as trapped-ion systems~\cite{Gaerttner2017, Richerme2014}. For Hamiltonians that lack a natural inversion, the effective sign flip can often be achieved using simple local Pauli conjugations~\cite{Monroe21RMP}. Fourth, the choice of measurement POVM $M$ depends on the chosen Hamiltonian $H$. However, for most Hamiltonians, measuring $m$ pairs in the computational basis suffices. Choosing a larger value of $m$ can increase the fidelity of the retained states,
but generally reduces the yield. Finally, additional error-correction steps may be applied after the protocol to further enhance performance. For instance, when the distilled entangled states are used for quantum key distribution, standard one-way hashing protocols~\cite{BDSW} can be employed to extract secret keys via classical post-processing. However, such one-way error correction typically relies on large CSS codes, whose coherent implementation remains experimentally challenging on current platforms. As a result, except in quantum key distribution settings, error correction is rarely applied in present-day experiments. Throughout our analysis, we assume that noise acts only on the input EPR pairs prior to distillation, while the Hamiltonian evolution itself is noiseless.

It is worth noting that the concrete performance targets for entanglement distillation depend on the intended application. Theoretical studies of asymptotic distillation and quantum cryptography typically focus on the yield of pure EPR pairs and the maximum tolerable noise threshold, whereas near-term quantum network protocols often require only that the output fidelity exceed a task-specific threshold---such as $F > 2/3$ for quantum teleportation~\cite{Popescu95, Ma12teleportation}, $F > (1 + 3/\sqrt{2})/4 \approx 0.78$ for Bell-inequality violation under a Werner-state model~\cite{QuantumEntanglement}, or $F \gtrsim 0.9$ for distributed quantum computing~\cite{Wehner2018}. Across both regimes, however, analog error detection serves as the foundational quantum step: in near-term applications, a single round of error detection can already boost the fidelity past the required operational threshold without coherent error correction, while in high-noise asymptotic regimes, it acts as a front-end filter that suppresses errors into the range where subsequent one-way error correction or classical post-processing succeeds. We therefore examine these two regimes in turn, analyzing standalone analog error detection in Section~\ref{sec:main} before turning to entanglement-based key distribution and noise tolerance in Section~\ref{sec:tolerance}.

\section{Error detection via Hamiltonian twirling}\label{sec:main}

In near-term implementations where coherent one-way error correction is omitted, the distillation protocol of Box~\ref{box:HEDP} relies entirely on its analog error-detection stage: evolving each $n$-pair block under bilateral Hamiltonian dynamics and post-selecting on matching outcomes across $m$ measured test pairs. The central theoretical question is what physical property of the Hamiltonian $H$ determines the effectiveness of this post-selection. If the pairs within a block do not interact prior to measurement, local measurements on $m$ sacrificed pairs probe only the measured subset; any low-weight Pauli error confined to the $n-m$ retained pairs commutes with the test observables and escapes detection altogether.

Hamiltonian scrambling overcomes this limitation by spreading initially localized errors across the entire $n$-qubit register. An input Pauli error $P$ acting on Alice's subsystem prior to the bilateral evolution $U(t)\otimes U^*(t)$, with $U(t)=e^{-iHt}$, is equivalent to the time-evolved operator $\tilde{P}(t)=U(t)PU^{\dagger}(t)$ acting immediately before measurement. Under an interacting many-body Hamiltonian, a low-weight error $P$ evolves into a superposition of high-weight Pauli operators supported across many qubits. Because the bilateral evolution $U(t)\otimes U^*(t)$ leaves the error-free state $\ket{\Phi^+}^{\otimes n}$ strictly invariant, ideal EPR pairs always produce correlated measurement outcomes. By contrast, once $\tilde{P}(t)$ spreads onto the $m$ test pairs and fails to commute with the local measurement observables, it disrupts the cross-party outcome correlations and triggers rejection of the block. The error-detection capability of a Hamiltonian is therefore governed by how effectively its dynamics generate non-commutativity between initial Pauli errors and the test-qubit measurements.

This physical mechanism is naturally quantified by out-of-time-order correlators (OTOCs), which track the growth of non-commutativity under Heisenberg evolution. For operators $V$ and $W$ on an $n$-qubit Hilbert space of dimension $d=2^n$, the infinite-temperature OTOC under $U(t)$ is defined as
\begin{equation}
    \mathrm{OTOC}_{V,W}^{U(t)}=\frac{1}{d}\tr\!\left(\tilde{V}^{\dagger}W^{\dagger}\tilde{V}W\right),\qquad
    \tilde{V}=U(t)VU^{\dagger}(t).
\end{equation}
We use the forward-conjugation convention, so that the dressed error agrees with the bilateral evolution in Box~\ref{box:HEDP}; relative to the usual Heisenberg convention this amounts to replacing $t$ by $-t$. For commuting Pauli observables, this correlator equals $1$, and its decay directly measures the non-commutativity generated by operator spreading. To connect this quantity to our distillation protocol, consider an input Pauli channel $\Lambda(\rho)=\sum_{P\in\mathsf{P}}c_P P\rho P$ acting on Alice's subsystem, with error-free weight $c_I>0$. Alice and Bob measure $1\leq m<n$ test pairs using conjugate POVMs $M=\{M_x\}$ and $M^*=\{M_x^*\}$ that accept the ideal EPR state with certainty,
\begin{equation}\label{eq:ideal_acceptance}
    \sum_x\tr\!\left[(M_x\otimes M_x^*)\ketbra{\Phi^+}{\Phi^+}^{\otimes n}\right]=1,
\end{equation}
with identity operators understood on the $n-m$ unmeasured pairs. Averaging the OTOC between each nonidentity Pauli error $P$ and the measurement operators $M_x$ over the noise distribution and the evolution-time distribution $\mu$, we define the averaged error-acceptance correlator
\begin{equation}\label{eq:averaged_OTOC}
    \mathrm{OTOC}_{\Lambda,M}^{\mathcal{N}_H}
    =\int_t\sum_{P\in\mathsf{P}\backslash I}\sum_x
    c_P\,\mathrm{OTOC}_{P,M_x}^{e^{-iHt}}\mu(t).
\end{equation}
The following theorem expresses the yield exactly through this averaged OTOC and gives a lower bound on the retained-block fidelity. Here $f_{\mathrm{out}}$ is the overlap of the accepted $(n-m)$-pair state with $\ket{\Phi^+}^{\otimes(n-m)}$.

\begin{theorem}[Distillation and OTOCs]\label{thm:OTOC}
    For Pauli noise $\Lambda$ and conjugate measurements satisfying Eq.~\eqref{eq:ideal_acceptance}, Hamiltonian distillation gives
    \begin{equation}
    \begin{split}
        f_{\mathrm{out}}&\geq\frac{c_I}{c_I+\mathrm{OTOC}_{\Lambda,M}^{\mathcal{N}_H}},\\
        Y&=\bigl(c_I+\mathrm{OTOC}_{\Lambda,M}^{\mathcal{N}_H}\bigr)(1-m/n),
    \end{split}
    \end{equation}
    where the averaged OTOC is defined in Eq.~\eqref{eq:averaged_OTOC}.
\end{theorem}

The theorem makes the error-detection picture quantitative: at fixed $c_I$, a smaller averaged OTOC means fewer initially erroneous components pass the check and gives a stronger fidelity guarantee. The bound need not be saturated: an accepted error can be confined to discarded pairs and leave the retained pairs perfect. Appendix~\ref{appendix:benign_errors} gives the exact nonnegative numerator correction. At fixed $m$, rejecting more blocks also lowers the yield. Stronger scrambling can improve this tradeoff by making errors detectable with fewer measured pairs, leaving a larger fraction of each accepted block available as output.

A simple choice of the POVM in Theorem~\ref{thm:OTOC} is to measure each of the $m$ selected qubits locally in the computational basis, on both Alice's and Bob's sides. Take the evolution time $t$ to be fixed. For a Pauli error $P$, define the dressed error operator by 
\begin{equation}\label{eq:dressed_P}
    \tilde{P}(t)=e^{-iHt}Pe^{iHt}
\end{equation}
label the measured qubits $1,\ldots,m$ and define the measured $Z$-type strings by
\begin{equation}\label{eq:Pauli_Z}
    Z^s=Z_1^{s_1}\cdots Z_m^{s_m},\quad s\in\{0,1\}^m,
\end{equation}
Here $Z^s$ applies $Z$ to the measured qubits selected by $s_j=1$ and identity elsewhere. Expanding the $m$-qubit computational-basis projectors in the Pauli-$Z$ basis reduces the averaged OTOC in Theorem~\ref{thm:OTOC} to a sum over $Z$-type strings on the measured subset as follows

\begin{corollary}[Local computational-basis measurements]\label{cor:Z_OTOC}
    For Pauli noise $\Lambda(\rho)=\sum_Pc_P P\rho P$ and local computational-basis measurements on $m$ pairs, the total probability of a nonidentity error passing the check at time $t$ is
    \begin{equation}\label{eq:Z_acceptance}
        p_{\mathrm{err}}(t)=\frac{1-c_I}{2^m}
        +\frac{1}{d\,2^m}\sum_{P\in\mathsf{P}\backslash I}c_P
        \sum_{s\neq 0}\tr\!\left[\tilde{P}^{\dagger}(t)Z^s\tilde{P}(t)Z^s\right],
    \end{equation}
    where $d=2^n$. Define the contribution of accepted nonidentity errors to ideal retained-state overlap by
    \begin{equation}\label{eq:benign_weight}
        b_{\Lambda,M}(t)=\sum_{P\in\mathsf{P}\backslash I}c_P
        \sum_{s\neq 0}\left|\frac{1}{d}\tr\!\left[Z^s\tilde P(t)\right]\right|^2.
    \end{equation}
    Then the output fidelity and yield are
    \begin{equation}\label{eq:local_exact_fidelity}
    \begin{split}
        f_{\mathrm{out}}&=\frac{c_I+b_{\Lambda,M}(t)}{c_I+p_{\mathrm{err}}(t)}
        \geq\frac{c_I}{c_I+p_{\mathrm{err}}(t)},\\
        Y&=\bigl(c_I+p_{\mathrm{err}}(t)\bigr)(1-m/n).
    \end{split}
    \end{equation}
    Here $0\leq b_{\Lambda,M}(t)\leq p_{\mathrm{err}}(t)$. For random evolution time, average $b_{\Lambda,M}(t)$ and $p_{\mathrm{err}}(t)$ separately over $\mu$ before taking their ratio.
\end{corollary}

The corollary gives a criterion for effective distillation using only local $Z$ measurements: the Hamiltonian must make the trace correction small compared with the first term in Eq.~\eqref{eq:Z_acceptance}. In this regime, only a fraction $2^{-m}$ of erroneous components pass, while ideal EPR pairs always survive. Each additional measured pair therefore halves the surviving error contribution, allowing simple local readout to achieve strong error rejection.

A natural model of strong scrambling is the Gaussian unitary ensemble (GUE), consisting of random Hermitian Hamiltonians with Gaussian matrix elements. These Hamiltonians mix configurations across the entire $n$-qubit system, allowing an initially low-weight error to develop components acting on many qubits. Such spreading makes errors more likely to disrupt the correlations tested by local measurements. The GUE therefore provides a useful benchmark for how effectively scrambling can support distillation. The following theorem gives sharp fidelity and yield bounds for the ensemble-averaged, post-selected state.

\begin{theorem}[GUE distillation]\label{thm:Haar}
    Let $H$ be an $n$-qubit normalized GUE Hamiltonian, and let the input Pauli channel have identity weight $c_I>0$. Average the long-time Hamiltonian twirling channel over $H$ before post-selection, and measure $1\leq m<n$ pairs in the computational basis. The output fidelity satisfy
    \begin{equation}\label{eq:GUE_sharp_fidelity_bound}
        f_{\mathrm{out}}\geq
        \frac{c_I+(1-c_I)\frac{(2^m-1)(d+4)}{d(d+1)(d+3)}}{c_I+(1-c_I)\left(2^{-m}+\frac{1-2^{-m}}{(d+1)(d+3)}\right)},
    \end{equation}
    and the output yield satisfies
    \begin{equation}\label{eq:GUE_yield_bounds}
        Y\geq \left[c_I+(1-c_I)\left(2^{-m}-\frac{1+2^{-m}+4/d}{(d+1)(d+3)}\right)\right](1-m/n).
    \end{equation}
    All three bounds are sharp at fixed $c_I$.
\end{theorem}

Under the large $d$ limit, the output fidelity is at least
\begin{equation}
    f_{\mathrm{out}}\geq\frac{c_{I}}{c_{I}+(1-c_{I})2^{-m}}
\end{equation}
where the error components are exponentially suppressed. Appendix~\ref{appendix:GUE_exact_fidelity} derives the exact fidelity and yield and proves the bounds above. The following corollary gives the threshold for exponentially small identity weight.

\begin{corollary}[Identity-weight threshold]\label{cor:GUE_threshold}
    Under the averaging prescription of Theorem~\ref{thm:Haar}, suppose $m=\rho n$ and $c_I=\kappa2^{-\lambda n}$ as $n\to\infty$, where $0<\rho<1$ and $\kappa,\lambda>0$. Then
    \begin{equation}\label{eq:GUE_threshold_corollary}
        f_{\mathrm{out}}\longrightarrow
        \begin{cases}
            1,&\lambda<\rho,\\
            \displaystyle\frac{\kappa}{1+\kappa},&\lambda=\rho,\\
            0,&\lambda>\rho.
        \end{cases}
    \end{equation}
\end{corollary}

In particular, exponentially small $c_I$ is compatible with asymptotically unit fidelity when $\lambda<\rho$. The yield obeys $Y\sim(1-\rho)\bigl(\kappa2^{-\lambda n}+2^{-\rho n}\bigr)$, reducing to $Y\sim(1-\rho)c_I$ in this regime. Appendix~\ref{appendix:GUE_threshold} proves the corollary and discusses critical prefactors, measurement offsets, a fixed number of retained pairs, and depolarizing noise.

This threshold has a direct connection to the separability boundary of depolarized EPR pairs. For local i.i.d.\ depolarizing noise acting on one party's qubits, each pair has fidelity $F_0=1-3p/4$, so $c_I=F_0^n$ and $\lambda=-\log_2 F_0$. The condition $\lambda<\rho$ therefore becomes
\begin{equation}
    p<\frac{4}{3}\left(1-2^{-\rho}\right).
\end{equation}
As the measured fraction $\rho$ approaches one, this threshold approaches $p=2/3$, precisely the point $F_0=1/2$ at which the noisy input pairs become separable~\cite{BDSW}. Thus, for every fixed $0<p<2/3$, a sufficiently large measured fraction $\rho<1$ yields asymptotically unit block fidelity, at the cost of an exponentially small yield. For $p\geq2/3$, we instead have $\lambda\geq1>\rho$, consistently excluding this high-fidelity regime: the input is separable, and local operations and classical communication cannot produce entanglement even with post-selection.

While Theorem~\ref{thm:Haar} relies on long-time averaging, practical implementations require error detection to develop within a short, finite coherence time. Our next theorem proves that a single GUE Hamiltonian achieves the same $2^{-m}$ leading error suppression at a fixed $O(1)$ evolution time independent of the system size $n$.

\begin{theorem}[Constant-time GUE distillation]\label{thm:O1}
    Consider $n$ noisy EPR pairs under a Pauli channel with error-free weight $c_I>0$, evolved under an $n$-qubit Hamiltonian $H$ drawn from normalized GUE for the fixed duration $\tau_1=j_{1,1}/2\approx 1.916$, where $j_{1,1}$ is the first positive zero of the Bessel function $J_1$. Measuring $1\leq m<n$ pairs locally in the computational basis yields, with overwhelming probability over $H$,
    \begin{equation}\label{eq:thm3_main}
        f_{\mathrm{out}}
        \geq\frac{c_I}{c_I+(1-c_I)\bigl(2^{-m}+O_{\prec}(d^{-1})\bigr)},\qquad
        Y
        =\Bigl[c_I+(1-c_I)\bigl(2^{-m}+O_{\prec}(d^{-1})\bigr)\Bigr]\!\left(1-\frac{m}{n}\right),
    \end{equation}
    where $X=O_{\prec}(d^{-1})$ denotes a remainder satisfying $\Pr_{H\sim\mathrm{GUE}}(|X|>d^{-1+\epsilon})\leq d^{-K}$ for every $\epsilon,K>0$ and sufficiently large $d$.
\end{theorem}

The constant timescale $\tau_1$ arises from destructive interference among the dynamical phases of the GUE spectrum. As $d\to\infty$, the empirical eigenvalue density converges to the Wigner semicircle law on $[-2,2]$, whose Fourier transform $\varphi(t)=J_1(2t)/t$ vanishes at $\tau_1=j_{1,1}/2$. At this point, the leading contributions to all non-trivial Pauli OTOCs in Eq.~\eqref{eq:Z_acceptance} cancel simultaneously, suppressing the error-acceptance probability to $2^{-m}+O_{\prec}(2^{-n})$ without long-time averaging. Whereas GUE Hamiltonians feature all-to-all couplings at fixed spectral width, the finite-time performance of spatially local trapped-ion and Rydberg Hamiltonians is demonstrated numerically in Sec.~\ref{sec:simulation}.

\section{Entanglement-based key distribution and noise tolerance}\label{sec:tolerance}

Beyond boosting state fidelity for direct quantum information tasks, entanglement distillation plays a foundational role in quantum cryptography through entanglement-based quantum key distribution (QKD)~\cite{BDSW, ShorPreskill2000, gottesman2003two-way, devetak2005distillation}. In an entanglement-based QKD scheme, Alice and Bob purify shared noisy EPR pairs before measuring them in the computational basis to extract a shared secret key. Within this setting, the performance of the underlying distillation protocol is characterized by two primary metrics. First, the asymptotic secret key rate is directly equivalent to the asymptotic distillation yield of pure EPR pairs. Second, the protocol's resilience in high-noise regimes is quantified by its \emph{noise tolerance}: the maximum initial channel error rate at which a positive asymptotic key rate (or nonzero distillation yield) can be sustained. High noise tolerance is especially vital for long-distance quantum communication, where fiber attenuation and detector dark counts severely degrade the initial entanglement.

Achieving a positive asymptotic key rate---or distilling asymptotically pure EPR pairs---requires supplementing the two-way error-detection stage with a one-way error-correction stage, such as one-way hashing (Step~(5) of Box~\ref{box:HEDP}). As discussed in Sec.~\ref{sec:protocol}, executing coherent one-way hashing directly on quantum registers requires large CSS circuits that remain challenging for near-term hardware. Entanglement-based QKD offers a natural way around this obstacle: because the surviving qubits are measured in the computational basis to generate raw keys, the one-way hashing operations commute with the key-generation measurements and reduce entirely to classical post-processing~\cite{ShorPreskill2000, gottesman2003two-way, Liu2022classically}. Specifically, bit-error correction translates into classical information reconciliation, while phase-error correction translates into privacy amplification. On its own, however, standard one-way classical post-processing can tolerate only modest channel noise, failing once the depolarizing bit and phase error rates exceed approximately $11\%$. By prepending our analog Hamiltonian error-detection stage prior to key generation, Alice and Bob can filter out heavily corrupted blocks using native analog dynamics and delegate the remaining error correction to classical post-processing.

To demonstrate the power of this hybrid analog-detection and classical-post-processing paradigm, we show that even a simple diagonal Hamiltonian $H_d=\mathrm{diag}(\lambda_1,\ldots,\lambda_{2^n})$ with distinct eigenvalues and non-degenerate eigenvalue gaps suffices to saturate the fundamental theoretical limit of noise tolerance. The key intuition comes from the asymmetric roles of bit and phase errors in classical post-processing: in the asymptotic secret fraction $1 - H(e_b) - H(e_p)$, once the residual bit error rate $e_b$ is suppressed to zero ($H(e_b)\to 0$), privacy amplification can tolerate any residual phase error rate $e_p < 1/2$ while maintaining a positive key rate $1 - H(e_p) > 0$. A diagonal Hamiltonian naturally engineers this division of labor between quantum error detection and classical post-processing. Under the bilateral evolution $e^{-iH_d t}\otimes e^{iH_d t}$, any state without bit-flip ($X$) errors remains within the correlated subspace spanned by $\ket{ii}$, where the dynamical phases cancel identically and quantum coherence is preserved. Conversely, any Pauli error $P=X_0 Z_0$ containing a nontrivial bit-flip component ($X_0 \neq I$) maps $\ket{ii}$ into off-diagonal basis states $\ket{ij}$ ($i\neq j$) that evolve at distinct nonzero frequencies $\lambda_i - \lambda_j$. Under long-time Hamiltonian twirling, these off-diagonal coherences average to zero, completely dephasing any block affected by a bit-flip error into a classical mixture of computational-basis strings (see Appendix~\ref{appendix:QKD_tolerance} for the detailed channel derivation).

In our QKD instantiation of Box~\ref{box:HEDP}, Alice and Bob apply this diagonal-Hamiltonian twirling and measure $m$ test pairs in the mutually unbiased Hadamard basis $\{\ket{+},\ket{-}\}$ for error detection, before measuring the surviving $n-m$ pairs in the computational basis to generate raw keys. Because a completely dephased state produces independent, uniformly random outcomes in the Hadamard basis, any block carrying a bit-flip error passes the $m$-pair consistency check with exponentially small probability,
\begin{equation}\label{eq:diagonal_bit_acceptance}
    \Pr(\text{all $m$ outcomes agree}\mid X_0\neq I)=\left(\frac{1}{2}\right)^m,
\end{equation}
whereas error-free EPR pairs always yield matching outcomes. Meanwhile, pure phase ($Z$) errors commute with $H_d$ and remain unaffected by the evolution; those falling on the $m$ test pairs are rejected by the Hadamard-basis check, while those on the $n-m$ retained pairs are removed classically via privacy amplification after key generation.

To quantify the resulting key rate and noise tolerance, consider local i.i.d.\ depolarizing noise of strength $p$ acting on one party's qubits, where on each pair the identity occurs with probability $1-3p/4$ and each Pauli error $X,Y,Z$ occurs with probability $p/4$. In the long-time twirling limit, the block acceptance probability $p_{\mathrm{pass}}$ and the conditional phase error rate $e_p$ of the retained pairs are given by
\begin{equation}\label{eq:QKD_parameters}
\begin{split}
    p_{\mathrm{pass}}&=\left(1-\frac{3p}{4}\right)^m\left(1-\frac{p}{2}\right)^{n-m}
    +2^{-m}\left[1-\left(1-\frac{p}{2}\right)^n\right],\\
    e_p&=\frac{\frac{p}{4-2p}\left(1-\frac{3p}{4}\right)^m\left(1-\frac{p}{2}\right)^{n-m}
    +2^{-(m+1)}\left[1-\left(1-\frac{p}{2}\right)^n\right]}{p_{\mathrm{pass}}}.
\end{split}
\end{equation}
Here, the two terms in $p_{\mathrm{pass}}$ correspond to accepted blocks without and with bit-flip errors, respectively. A block free of bit-flip errors passes the Hadamard check if and only if all $m$ measured pairs are completely error-free, leaving each of the $n-m$ unmeasured pairs with a conditional phase error probability of $(p/4)/(1-p/2)=p/(4-2p)$. Conversely, a block containing at least one bit-flip error passes with probability $2^{-m}$ by Eq.~\eqref{eq:diagonal_bit_acceptance} and contributes a phase error rate of $1/2$ due to complete dephasing. Averaging these two contributions over the accepted blocks yields $e_p$.

The corresponding residual bit error rate among the retained pairs is $e_b=2^{-m}(p/2)/p_{\mathrm{pass}}$. For finite block sizes, classical information reconciliation consumes an entropy fraction $H(e_b)$ per retained pair. In the asymptotic regime where $m$ and $n$ grow at fixed ratio $m/n$ such that $e_b$ is exponentially suppressed, this reconciliation cost vanishes and the secret fraction per retained pair simplifies to $1-H(e_p)$, where $H(\cdot)$ denotes the binary entropy. This leads directly to our analytical key-rate and noise-tolerance theorem.

\begin{theorem}[Key rate and noise tolerance]\label{lemma:diagonal+classical}
    For the diagonal-Hamiltonian QKD protocol above, in the asymptotic regime of exponentially suppressed bit errors, the secret key rate is
    \begin{equation}
        R=p_{\mathrm{pass}}\bigl(1-H(e_p)\bigr)(1-m/n),
    \end{equation}
    where $H(\cdot)$ is the binary entropy. The maximum tolerable depolarizing strength approaches the threshold $p_{\mathrm{tol}}\in[0,2/3]$ determined by
    \begin{equation}\label{eq:tolerance}
        \left(1-\frac{p_{\mathrm{tol}}}{2}\right)^{1-m/n}
        \left(2-\frac{3p_{\mathrm{tol}}}{2}\right)^{m/n}=1.
    \end{equation}
\end{theorem}

Theorem~\ref{lemma:diagonal+classical} shows that the noise tolerance is governed by the fraction $m/n$ of pairs sacrificed for analog error detection. In particular, maximizing the error-detection fraction by measuring $n-1$ pairs and retaining a single pair per block yields the highest possible threshold:

\begin{corollary}\label{thm:tolerance}
    As $m/n\to 1$, the diagonal-Hamiltonian protocol approaches the depolarizing threshold $p_{\mathrm{tol}}=2/3$, corresponding to a channel bit and phase error rate of $1/3 \approx 33.3\%$.
\end{corollary}

Because a two-qubit Werner state subjected to local depolarizing noise of strength $p=2/3$ (an isotropic Pauli error rate of $33.3\%$) becomes separable, Corollary~\ref{thm:tolerance} saturates the fundamental theoretical upper bound: our hybrid analog-detection and classical-post-processing scheme can extract a secret key---or, equivalently, distill pure entanglement when paired with one-way hashing---whenever the initial noisy states retain any nontrivial entanglement.

\section{Simulation}\label{sec:simulation}

In the previous sections, we established theoretical results for general Hamiltonians, demonstrating their efficacy in constructing entanglement distillation protocols. However, real-world Hamiltonians are typically constrained to specific families. To bridge the gap between theory and practice, we present numerical simulations using representative device-native Hamiltonians to demonstrate that our protocol maintains high performance under these experimentally relevant models and requires only experimentally feasible evolution times. We also discuss the impact of skipping Pauli twirling in the presence of non-Pauli noise channels. Finally, we provide additional numerical results on noise tolerance and representative applications including quantum key distribution and quantum repeaters.

\subsection{Performance of Device-native Hamiltonian}

The Hamiltonian entanglement distillation protocol is broadly applicable to a variety of quantum systems. Here, we illustrate its application using two representative device-native Hamiltonians: those of Rydberg atom systems and trapped ion systems. In a Rydberg system, a qubit is encoded using the ground state $\ket{g}$ and the Rydberg state $\ket{r}$ of a neutral atom. The system Hamiltonian can be written as
\begin{equation}
    H_{rydberg}=\frac{\Omega}{2}\sum_{j}(e^{i\phi}\ket{g_j}\bra{r_j}+h.c.)-\Delta\sum_{j}\hat{n}_j+\sum_{j<k}V_{jk}\hat{n}_j\hat{n}_k,
\end{equation}
where $\hat{n}=\ket{r}\bra{r}$. Following the parameters reported in~\cite{BloqadeHamiltonians}, we set the driving frequency to $\Omega = 2\pi \times 1~\mathrm{MHz}$ and the detuning to $\Delta = 2\pi \times 2.5~\mathrm{MHz}$. The interaction strength is given by $V_{jk} = C/|x_j - x_k|^6$, with $C = 2\pi \times 862890~\mathrm{MHz}\cdot\mu\mathrm{s}^6$. For simplicity, we assume the atoms are arranged in a linear chain with uniform spacing of $6~\mu\mathrm{m}$, slightly below the blockade radius, so that the interaction strength could be stronger.

For trapped-ion systems, effective spin–spin interactions can be engineered to realize long-range interaction Hamiltonians of the form~\cite{Monroe21RMP}
\begin{equation}
H_{\rm ion} = \sum_{i<j} J_{ij} \hat{\sigma}_i^x \hat{\sigma}_j^x + B \sum_i \hat{\sigma}_i^z.
\end{equation}
where the couplings are well approximated by a power-law decay $J{ij} = J_0/|i-j|^{\alpha}$, with tunable exponent $\alpha$ typically ranging from $0$ to $3$. We assume the ions are arranged in a uniform linear chain and take the nearest-neighbor coupling strength to be $J_0 = 2\pi \times 1~\mathrm{kHz}$. The transverse field is set to $B = 2J_0$, and throughout our simulations we fix $\alpha = 1$, corresponding to long-range interactions. In addition, we include two idealized Hamiltonians for comparison. The first is the transverse-field Ising model with periodic boundary conditions, given by:
\begin{equation}
\tilde{H}_{TFIM} = J_0 \sum_{i} \hat{\sigma}_i^x \hat{\sigma}_{i+1}^x + B \sum_i \hat{\sigma}_i^z.
\end{equation}
with the same parameters $J_0$ and $B$ as in the trapped-ion case. This is a simplified model of trapped-ion Hamiltonian. The second is a diagonal Hamiltonian $H_{\rm diag} = \mathrm{diag}(\lambda_1,\ldots,\lambda_d)$ satisfying the non-degenerate eigenvalue and gap conditions. Owing to the absence of off-diagonal terms, this Hamiltonian exhibits minimal scrambling and therefore represents the weakest case for error detection in our protocol.

For our simulations, unless stated otherwise, we consider Hamiltonians acting on $n=5$ qubit pairs, with error detection performed by measuring $m=3$ pairs in the Hadamard basis, which yields nontrivial performance for diagonal Hamiltonians. We assume a noise model of local i.i.d. depolarizing noise at $p=0.2$. The performance of different Hamiltonians are compared through yield and fidelity. We adopt the normalized block fidelity $\overline{f} = f_{\text{out}}^{1/(n-m)}$ as fidelity metric, as used in Ref.~\cite{gu2025}. This normalization enables consistent comparisons between protocols producing varying numbers of entangled pairs.

We begin by demonstrating the asymptotic time performance of the protocol under four distinct Hamiltonian settings. Figure~\ref{fig:F-Y-m} shows the asymptotic fidelity and yield for diagonal, Ising, trapped-ion and Rydberg Hamiltonians as a function of the number of measured qubits. To show the scaling, we simulate Hamiltonians with $10$ qubits. We assume that one party of each EPR pair is subject to local depolarizing noise of strength $p = 0.2$, therefore the initial average fidelity of $\overline{f_0} = 1 - 3p/4 = 0.85$ before distillation. Ideally, increasing the number of measured qubits is expected to improve fidelity at the cost of a reduced yield. Overall, the trends observed in Fig.~\ref{fig:F-Y-m} are consistent with these expectations. The trapped-ion Hamiltonian and its periodic-boundary Ising approximation exhibit comparable performance. Although omitting the long-range interactions in the Ising model reduces its scrambling strength, the periodic boundary condition enhances qubit connectivity, partially compensating for this effect and resulting in similar distillation performance. The Rydberg Hamiltonian achieves smaller fidelity improvements compared to the trapped-ion case. This can be attributed to the faster spatial decay of interactions in Rydberg systems, which scale as $r^{-6}$, whereas the trapped-ion model considered here assumes $\alpha=1$, corresponding to a slower $r^{-1}$ decay and hence stronger long-range coupling. As a result, the trapped-ion Hamiltonian provides more effective error scrambling. The diagonal Hamiltonian yields the poorest performance among all Hamiltonians, due to its inability to detect and suppress $Z$-type. Compared with a single round of the recurrence method, the Hamiltonian-based protocol achieves substantially higher fidelity at the cost of reduced yield. This trend would also persist when both methods are applied over multiple rounds. Importantly, in realistic experimental settings, implementing multiple distillation rounds leads to great overhead due to the need for measurement feedback and the requirement to physically regroup qubits that may be spatially separated. Therefore, the ability of the Hamiltonian-based method to achieve higher fidelity within fewer rounds constitutes a significant practical advantage.

\begin{figure}[htbp]
    \centering
    \includegraphics[width=0.5\linewidth]{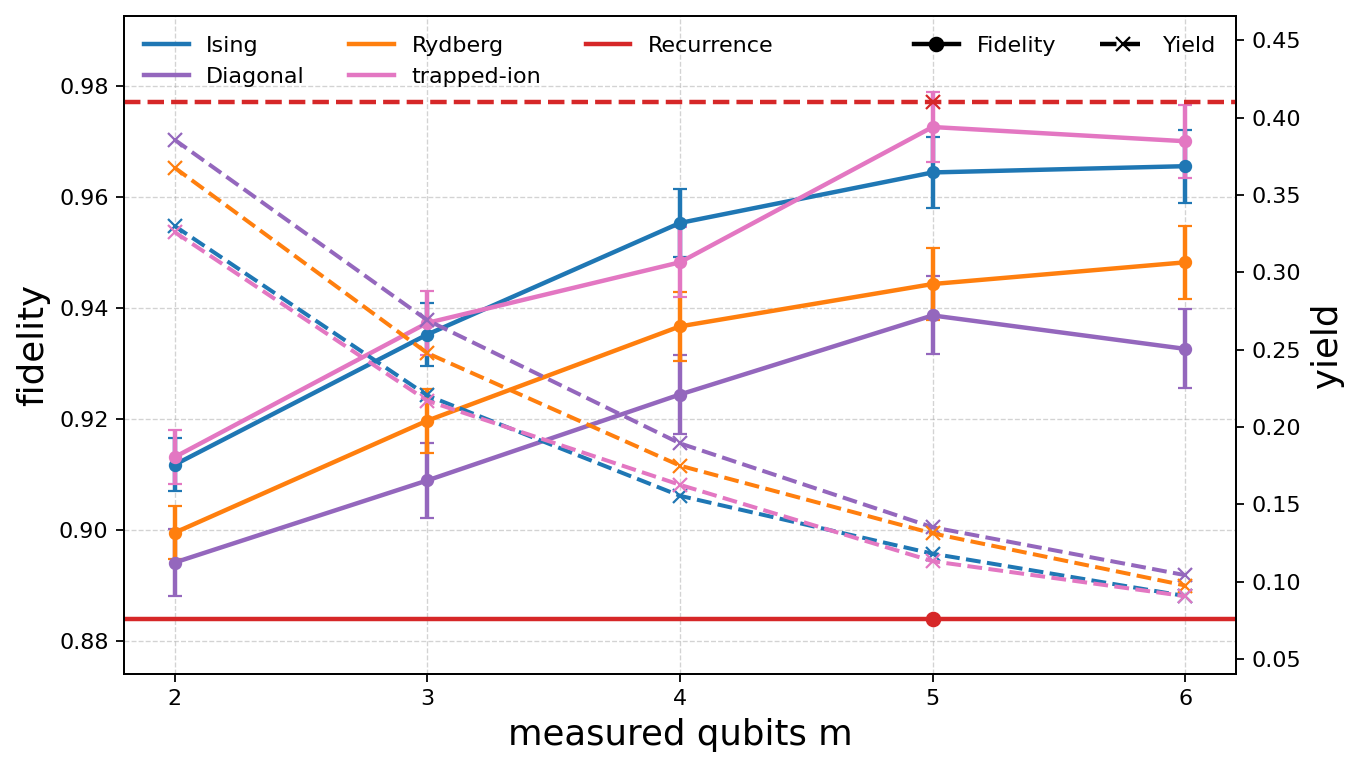}
    \caption{Distillation performance as a function of the number of measurement qubits. The solid line with circle markers denotes the output fidelity, while the dashed line with cross markers denotes the yield. We consider a ten-qubit Hamiltonian with local depolarizing noise of strength $p=0.2$.}
    \label{fig:F-Y-m}
\end{figure}

In Fig.~\ref{fig:F-Y-noise}, we examine the distillation performance under varying strengths of local depolarizing noise. The results show that the periodic-boundary Ising model achieves performance comparable to that of the trapped-ion Hamiltonian. Moreover, the trapped-ion Hamiltonian consistently outperforms the Rydberg Hamiltonian, which in turn outperforms the diagonal Hamiltonian across the entire noise range. Comparing this figure at point $p=0.2$ and the previous figure demonstrates that enlarging the Hamiltonian system and measuring more qubits substantially enhances distillation performance. From the above asymptotic analysis, we conclude that realistic Hamiltonians consistently outperform diagonal Hamiltonians across different measurement and noise settings. This observation implies that analytical results derived for diagonal Hamiltonians can be interpreted as conservative lower bounds for the performance of many practically relevant Hamiltonians, for which closed-form expressions are often difficult to obtain. Accordingly, in the following simulation sections, we focus on the more practically relevant trapped-ion and Rydberg Hamiltonians.

\begin{figure}[htbp]
    \centering
    \includegraphics[width=0.5\linewidth]{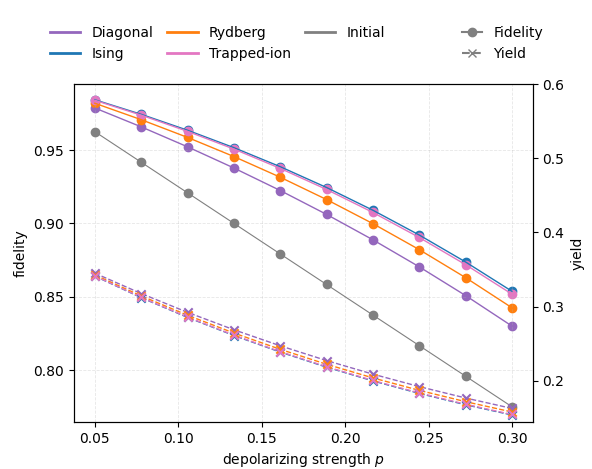}
    \caption{Distillation performance as a function of the depolarizing noise strength. The solid line with circle markers denotes the output fidelity, while the dashed line with cross markers denotes the yield. We consider a six-qubit Hamiltonian with three measurement qubits.}
    \label{fig:F-Y-noise}
\end{figure}

We now consider the finite-time effects of different Hamiltonians, which are crucial for realistic applications, as current quantum platforms have limited coherence times. The Rydberg platform operates at frequencies on the order of $\mathrm{MHz}$, while trapped-ion systems are typically in the $\mathrm{kHz}$ range. To enable a meaningful comparison, we define a unit of time as $1/2\pi~\mu\mathrm{s}$ for Rydberg atoms and $1/2\pi~\mathrm{ms}$ for trapped ions, allowing both time scales to be displayed on the same plot in Figure~\ref{fig:finite-time}. 
The asymptotic performance for both Hamiltonians can be approximated in Figure~\ref{fig:F-Y-noise}, where the trapped-ion Hamiltonian has fidelity of approximately $0.925$ and the Rydberg Hamiltonian reaches $0.91$.

\begin{figure}[htbp]
    \centering
    \includegraphics[width=0.5\linewidth]{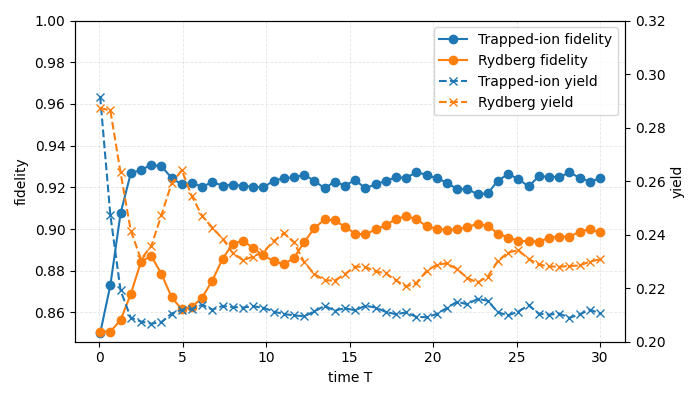}
    \caption{Finite-time distillation performance for trapped-ion and Rydberg Hamiltonians. The characteristic evolution times are at the microsecond ($\mu$s) scale for the Rydberg Hamiltonian and at the millisecond (ms) scale for the trapped-ion Hamiltonian. The solid line with circle markers denotes the output fidelity, while the dashed line with cross markers denotes the yield.}
    \label{fig:finite-time}
\end{figure}

Figure~\ref{fig:finite-time} shows that evolving the trapped-ion Hamiltonian for just 2 units of time (about $1/\pi~\mathrm{ms}$) is sufficient to achieve fidelity higher than asymptotic performance. For reference, experiments in~\cite{Guo2024} demonstrate coherence times of roughly $2~\mathrm{ms}$ over hundreds of qubits, indicating that Hamiltonian entanglement distillation is already practical on current devices. For the Rydberg Hamiltonian, fidelity reaches a plateau close to its asymptotic fidelity at about $13$ unit time(about $2\rm{\mu}s$) and then becomes quite stable in this high fidelity plateau. This again reflects the weaker scrambling power of the Rydberg Hamiltonian compared to trapped-ion Hamiltonian, requiring longer evolution times to approach its asymptotic fidelity. Both Hamiltonians exhibit oscillatory behavior in fidelity, with the effect being more significant for the Rydberg Hamiltonian. This behavior can be attributed to the presence of quasi-degenerate eigenvalue gaps, or equivalently, very small different in the eigenvalue gaps, which compromise efficient scrambling over time. Importantly, these results indicate that the Hamiltonian-based protocol remains effective in the presence of moderate spectral degeneracies and does not require a fully nondegenerate spectrum.  

\subsection{Non-Pauli Noises}\label{subsec:non-Pauli}
In our protocol in Sec~\ref{sec:protocol}, we include a step of Pauli twirling to convert arbitrary noise into Pauli noise. While this is a digital operation, it can still be implemented on many analog simulators. To support our claim that the protocol performs well even without this step, we analyze the effect of Pauli twirling on non-Pauli noise, demonstrating that a purely analog implementation remains effective. Specifically, we focus on amplitude damping noise, which is commonly encountered in trapped-ion and neutral-atom systems. This noise is both non-Pauli and non-unital, making it challenging to treat in many theoretical settings. In our simulations, we consider a local depolarizing noise with strength $p = 0.2$, followed by a local amplitude damping noise with varying strength $\gamma$. This allows us to investigate how the protocol performs in a realistic analog scenario without relying on Pauli twirling.

\begin{figure}[htbp]
    \centering
    \includegraphics[width=0.5\linewidth]{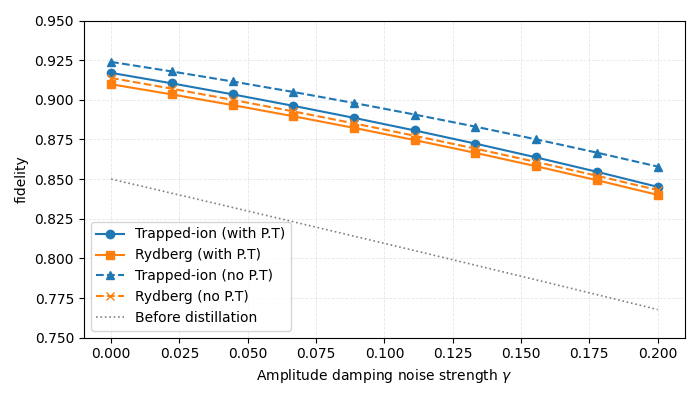}
    \caption{Distillation performance with and without Pauli twirling. The solid line with circle markers corresponds to the case with Pauli twirling, while the dashed line with triangle markers corresponds to the case without twirling.}
    \label{fig:Pauli_Twirling}
\end{figure}

The performance of different Hamiltonians, both with and without Pauli twirling, is shown in Figure~\ref{fig:Pauli_Twirling}. From Figure~\ref{fig:Pauli_Twirling}, we observe that for local amplitude damping noise, when measured in the hadamard basis, the protocol without Pauli twirling actually achieves higher fidelity performance than the Pauli twirling version. This is potentially because the Pauli Liouville representation of amplitude damping channel is 
\begin{equation}
    \mathcal{R}_{\text{AD}} = \begin{pmatrix}
        1 & 0 & 0 & 0 \\
        0 & \sqrt{1-\gamma} & 0 & 0 \\
        0 & 0 & \sqrt{1-\gamma} & 0 \\
        \gamma & 0 & 0 & 1-\gamma
    \end{pmatrix},
\end{equation}
in the Pauli basis 2$\{I, \sigma_x, \sigma_y, \sigma_z\}$ where additional $Z$ operators are introduced before Pauli twirling and thus being more probably to be detected through the X-basis measurement. Therefore even in the absence of Pauli twirling, the protocol remains effective, demonstrating its robustness for purely analog implementations.

\subsection{Noise tolerance in quantum key distribution and repeaters}

Finally, we present additional numerical results on noise tolerance. Figure~\ref{fig:tolerable_asymp} shows the noise tolerance for different fractions of measured qubits in the asymptotic qubits number $n$ limit. The blue region below the curve are the tolerable error region while gray region above would result in the protocol to output zero EPR pair. Measuring all but one qubit is represented as the rightmost point of the curve, which allows the protocol to approach the $1/3$ threshold, while measuring only half of the qubits reduces the tolerable error rate to $(5-\sqrt{13})/6 \approx 23\%$. The noise tolerance for finite $n$ and $m$ are shown in Figure~\ref{fig:tolerable_finite}. For a Hamiltonian of size $n = 20$, the maximum tolerable error rate is approximately $26\%$. When the qubit number increases to $n = 50$, the noise tolerance exceeds $30\%$. Note that when the measuring qubits is very few, the noise tolerance is quite low, even less than the  noise tolerance for one-way hashing protocol. This is because in the diagonal Hamiltonian twirling step, we would introduce more phase error for the states that already have $X$-error. So if the probability of those components are not well suppressed, we are introducing too much phase error and compromise the further one-way distillation step. To guarantee good noise tolerance and practical performance, we typically need to measure $O(n)$ qubits for error detection. 

\begin{figure}[htbp]
    \centering
    \subfigure[Asymptotic tolerable error rate.]{
        \includegraphics[width=0.52\linewidth]{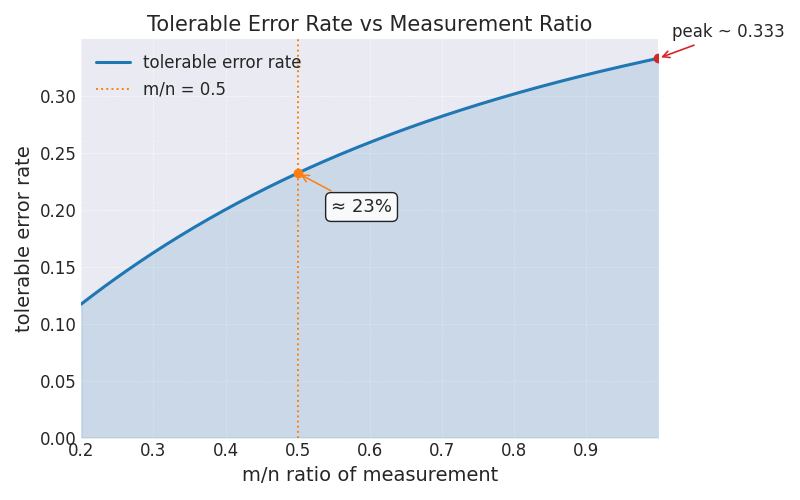}
        \label{fig:tolerable_asymp}
    }
    \hfill 
    \subfigure[Finite tolerable error rate.]{
        \includegraphics[width=0.44\linewidth]{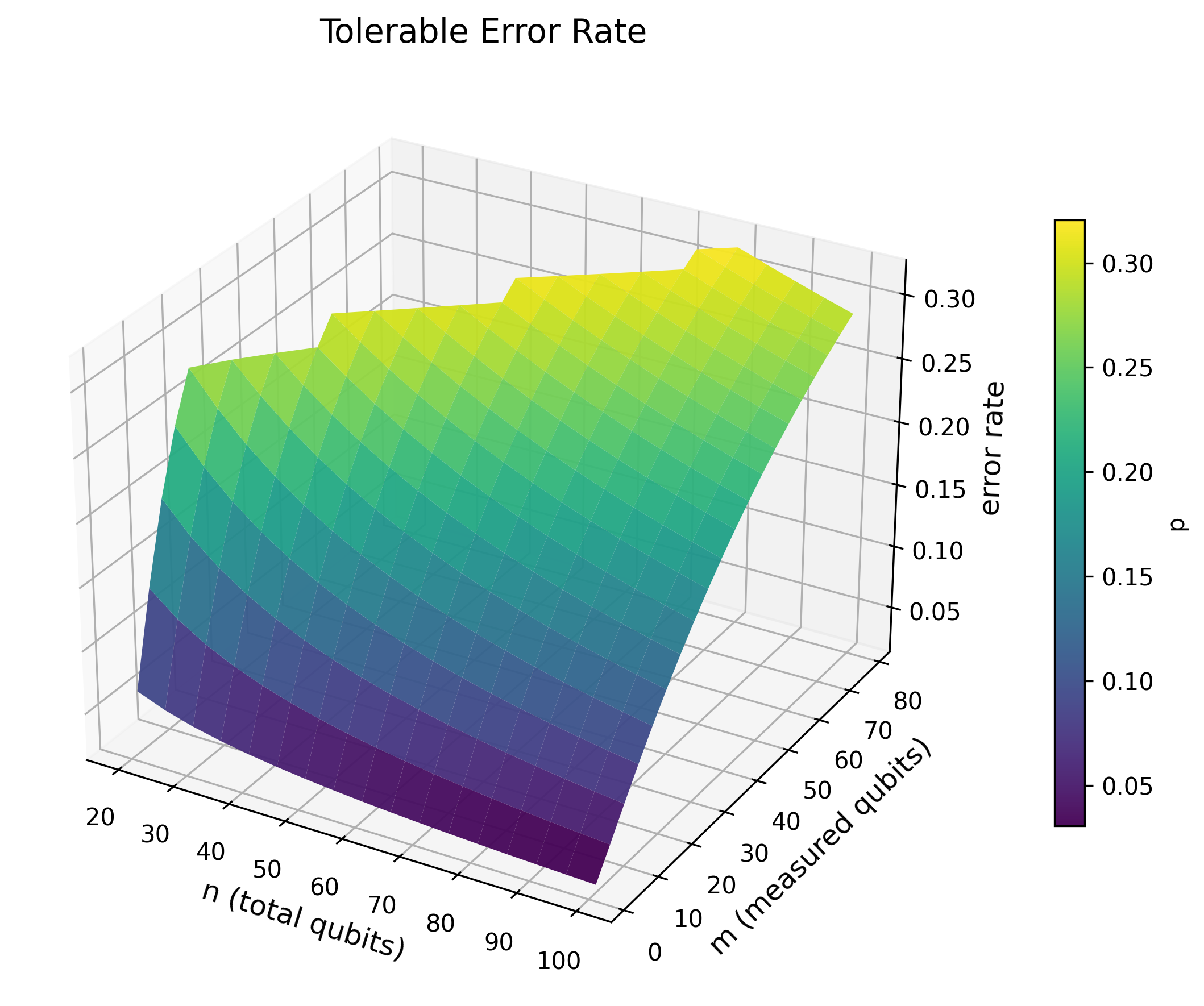}
        \label{fig:tolerable_finite}
    }
    \caption{Tolerable error rates under varying conditions: (a) asymptotic behavior as a function of the measurement qubit ratio $m/n$; (b) performance for finite system sizes where $n, m \leq 100$.}
    \label{fig:combined_plots}
\end{figure}

To illustrate the practical implications of noise tolerance in quantum networks, we consider two representative applications. The first is the maximum achievable transmission distance in quantum key distribution. The second concerns quantum repeater networks: specifically, the largest separation between neighboring repeaters for which the distilled entanglement can still exceed a certain fidelity threshold. Following recent work on quantum repeaters~\cite{Liu2026}, we set the fidelity threshold to $90\%$, which could lead to the violation of the Bell inequality up to approximately $2.5$. Further details on the channel noise model and simulation parameters are provided in the Appendix~\ref{appendix:QKD} and references~\cite{AQT,Liu2026}. Figure~\ref{fig:QKD} compares the maximum transmission distance of different QKD protocols versus the photon bit error rate $e_d$ determined by the quantum channel. Specifically, we use a protocol based on a $15$-qubit diagonal Hamiltonian, with error detection performed on $12$ qubits. Following the diagonal-Hamiltonian QKD scheme in Sec.~\ref{sec:tolerance}, Hamiltonian error detection suppresses bit-error components, and classical post-processing reconciles the residual bit errors and performs privacy amplification. We take the information-reconciliation efficiency to be $f=1.06$~\cite{Zhu24}. The Hamiltonian-based protocol attains a larger maximum distance, outperforming the recurrence protocol even after two rounds of iteration. This improvement reflects the enhanced noise tolerance enabled by Hamiltonian-based distillation. Figure~\ref{fig:repeater} shows the distilled fidelity versus transmission distance. Since no error correction could be applied in this case, we assume Hamiltonians draw from GUE instead of diagonal Hamiltonian in this setting. Note that when no distillation is performed, the fidelity drops below $90\%$ at about $120\mathrm{km}$, while two-way method can boost this $90\%$ fidelity threshold to about $300\mathrm{km}$, with Hamiltonian-based method achieving highest separation distance.

\begin{figure}[htbp]
    \centering
    \subfigure[Maximum transmission distance in QKD.]{
        \includegraphics[width=0.44\linewidth]{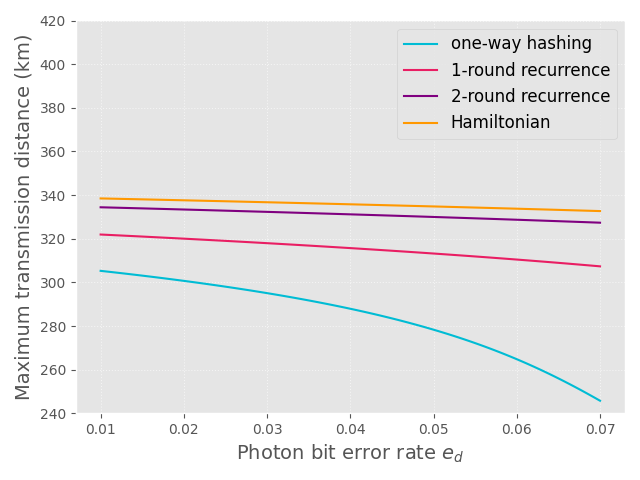}
        \label{fig:QKD}
    }
    \hfill 
    \subfigure[Largest separation between neighboring repeaters.]{
        \includegraphics[width=0.48\linewidth]{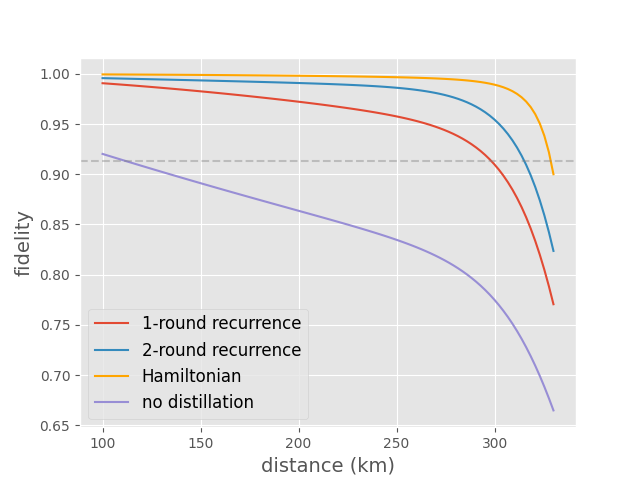}
        \label{fig:repeater}
    }
    \caption{Noise tolerance shown in practical application.}
    \label{fig:combined_plots}
\end{figure}

\section{Outlook}\label{sec:outlook} 

In this work, we introduced a class of entanglement distillation protocols based on Hamiltonian evolution. In contrast to conventional digital schemes, our approach uses the intrinsic scrambling dynamics of many-body Hamiltonians to enable efficient error detection and high-performance distillation. We established a quantitative connection between a Hamiltonian's scrambling capability and its distillation performance via out-of-time-order correlators, showing that generic interacting Hamiltonians provide sufficient scrambling for this task and that typical Gaussian unitary ensemble (GUE) Hamiltonians achieve near-optimal error detection at an $O(1)$ evolution time independent of system size. Furthermore, by combining analog error detection under a simple diagonal Hamiltonian with classical post-processing in entanglement-based QKD, we showed that our framework saturates the theoretical $33.3\%$ depolarizing noise-tolerance bound.

On the practical side, we validated our protocol through numerical simulations of experimentally relevant platforms, including Rydberg atom arrays and trapped-ion systems. The results demonstrate strong fidelity improvements within short, experimentally accessible evolution times, even in the absence of Pauli twirling under non-Pauli amplitude-damping noise. Additional simulations for quantum key distribution and quantum repeaters show that the high noise tolerance of our protocol enables longer maximum transmission distances than conventional few-round recurrence protocols. Because the protocol requires only evolution under a native Hamiltonian followed by local measurements, it avoids fine-grained pulse-level control, complex circuit compilation, and high-overhead quantum error correction, making it well suited for near-term analog quantum simulators.

Several directions for future work naturally arise. Throughout this study, we assumed that noise acts only on the input EPR pairs prior to distillation, while treating the Hamiltonian evolution itself as ideal and unitary. In realistic analog devices, however, qubits also experience decoherence and dissipation during the evolution. A natural next step is therefore to extend our analysis to open-system dynamics described by a Lindblad master equation, where dissipative terms compete directly with coherent scrambling. Determining the resulting trade-off between scrambling-based error detection and in-evolution dissipation, and characterizing the parameter regimes in which the protocol still yields a net fidelity improvement, is an important open question. In addition, while we established $O(1)$ distillation times for GUE Hamiltonians, it would be valuable to investigate finite-time performance for structured, spatially local, or time-dependent, e.g., driven or Floquet, Hamiltonians, where faster scrambling could further reduce exposure to dissipation during evolution. Finally, since quantum networks rely heavily on optical platforms, generalizing this Hamiltonian-based distillation framework to Bosonic systems remains an interesting direction for future research.

\acknowledgments
We thank Michael Gullans, Shiv Akshar Yadavalli, Chushi Qin, Xiaodi Wu, Yingkang Cao, Haohai Shi, Zhenyu Du for their helpful discussions. Z.X acknowledges support from Joint Center for Quantum Information and Computer Science (QuICS) Lanczos Graduate Fellowship. G.L. acknowledges the National Natural Science Foundation of China Grant No.~12575023. AI tools was used to refine the result of Theorem 2 and Theorem 3, as well as language editing and figure preparation. The authors checked all AI-assisted content and take full responsibility for the manuscript.

\nocite{forceTitles}
\bibliographystyle{apsrev4-2}

\bibliography{ref}

\clearpage

\appendix

\onecolumngrid
\begin{center}
\textbf{\large Supplementary Material }
\end{center}

The appendices are organized as follows. Appendix~\ref{appendix:preliminaries} reviews the notation, standard distillation protocols, twirling, OTOCs, and the GUE conventions used in the main text. Appendix~\ref{appendix:proofs} contains the proofs and performance refinements: Appendices~\ref{appendix:OTOC}--\ref{appendix:benign_errors} establish the OTOC acceptance formulas and the exact retained-state fidelity; Appendix~\ref{appendix:GUE_exact_fidelity} derives the finite-size GUE result of Theorem~\ref{thm:Haar}; and Appendix~\ref{appendix:GUE_threshold} proves Corollary~\ref{cor:GUE_threshold} and analyzes its threshold regimes. Appendices~\ref{appendix:O1} and~\ref{appendix:QKD_tolerance} prove the constant-time GUE and diagonal-Hamiltonian QKD results, respectively. Finally, Appendix~\ref{appendix:QKD} specifies the channel models and parameters used in the QKD and repeater simulations.

\section{Preliminaries}\label{appendix:preliminaries}
\subsection{Basic notations and quantities}
Let us start from the basic notations of this work. We denote the Hilbert space as $\mathcal{H}$. A quantum state $\rho$ is a Hermitian operator in the Hilbert space which has unity trace. When this state is a pure state, sometimes we will write it as a vector $\ket{\psi}$, which in terms meaning that the density matrix is $\ket{\psi}\bra{\psi}$. The physical operation $\mathcal{N}:\mathcal{H}\rightarrow\mathcal{H}$ that maps a quantum state to another quantum state is called a quantum channel, which should satisfy CPTP condition. Here CP means completely positive, that is for any positive operator $\sigma\in\mathcal{H}\otimes\mathcal{H}_0$, the channel should have $\mathcal{N}\otimes\mathcal{I}(\sigma)>0$ with $\mathcal{I}$ being the identity channel in Hilbert space $\mathcal{H}_0$. The measurement of a quantum state is performed by a positive operator valued measurement, or POVM $\{M_0,\cdots, M_{k}\}$ such that each $M_i\in\mathcal{H}$ is a positive operator and $\sum_{i}M_i=I$.

In this paper, certain quantum states, channels, and measurements will be frequently used. The most important states are the four Bell state on a bipartite system $\mathcal{H}_A\otimes\mathcal{H}_B$:
\begin{equation}
\begin{split}
\ket{\Phi^+} = \frac{1}{\sqrt{2}}\left( \ket{00} + \ket{11} \right),\\
\ket{\Phi^-} = \frac{1}{\sqrt{2}}\left( \ket{00} - \ket{11} \right),\\
\ket{\Psi^+} = \frac{1}{\sqrt{2}}\left( \ket{01} + \ket{10} \right),\\
\ket{\Psi^-} = \frac{1}{\sqrt{2}}\left( \ket{01} - \ket{10} \right).
\end{split}
\end{equation}
Among the four states, $\ket{\Phi^+}$, sometimes called EPR state, is of particular interest due to its symmetry. For example, it can be extended naturally into the $2n$-qubit EPR state
\begin{equation}
\ket{\Phi^+}^{\otimes n} = \frac{1}{\sqrt{2^n}}\sum_{i=0}^{2^{n}-1} \ket{ii},
\end{equation}
where the first $\ket{i}$ indicates the state in system $\mathcal{H}_A$ and the second $\ket{i}$ indicate the state in system $\mathcal{H}_B$. One important property of EPR state enabling our work is the following identity: 
\begin{equation}\label{eq:useful_identity}
    U_A\ket{\Phi^+}_{AB}=U_B^T\ket{\Phi^+}_{AB}.
\end{equation}
Together with the fact that $(U^T)^{\dagger}=U^{*}$, we can see that $U_AU_B^*\ket{\Phi^+}_{AB}=\ket{\Phi^+}_{AB}$. The EPR state serves as a fundamental resource in quantum information science. However, during the transmission of EPR states, noise is inevitable. Two commonly encountered noise models are depolarizing noise and dephasing noise. The single-qubit depolarizing channel with noise strength $p$ is defined as
\begin{equation}
\mathcal{D}_p(\rho)=(1-p)\rho+p\frac{I}{2},
\end{equation}
Since for arbitrary single qubit state $\rho$ we have $\frac{1}{4}(\rho+X\rho X+Y\rho Y+Z\rho Z)=\frac{I}{2}$. This channel can be also written as
\begin{equation}
\mathcal{D}_p(\rho)=(1-\frac{3}{4}p)\rho+\frac{p}{4}X\rho X+\frac{p}{4}Y\rho Y+\frac{p}{4}Z\rho Z,
\end{equation}
The above two definition are equivalent. This channel isotropically randomizes the state and models symmetric errors in all Pauli directions. Another useful channel is the dephasing channel with error rate $p$ is given by
\begin{equation}
\mathcal{Z}_p(\rho)
=
(1-p)\rho
+
p\, Z\rho Z.
\end{equation}
This noise preserves populations in the computational basis while suppressing off-diagonal coherence. In our main theorem, we will consider a special case of the dephasing channel, called global complete dephasing channel, defined by
\begin{equation}
    \Delta(\rho)=\mathcal{Z}_{1/2}^{\otimes n}(\rho).
\end{equation}
This channel removes all off-diagonal elements of a density matrix in the computational basis, effectively turning it into a classical mixture. 

Both depolarizing noise and dephasing noise belong to the class of \emph{Pauli channels}. A Pauli channel is a quantum channel whose action can be expressed as a probabilistic mixture of Pauli operators. For an $n$-qubit system, a Pauli channel takes the form of
\begin{equation}
\Lambda(\rho)
=
\sum_{P \in \mathsf{P}_n} c_P \, P \rho P,
\end{equation}
where $\mathsf{P}_n = \{I,X,Y,Z\}^{\otimes n}$ denotes the $n$-qubit Pauli group (up to phases), and $\{c_P\}$ is a probability distribution. Local depolarizing noise corresponds to the single qubit case in which all non-identity Pauli operators occur with equal probability, while dephasing noise corresponds to Pauli channels supported only on $Z$-type operators. Pauli channels are diagonal in the Pauli operator basis and naturally arise after Pauli twirling, making them a convenient and widely used noise model in quantum information theory. Suppose there is a Pauli channel $\Lambda_{AB}$ that acts on Hilbert space $\mathcal{H}_A\otimes\mathcal{H}_B$ with $\mathcal{H}_A, \mathcal{H}_B$ having equal dimension. Then there always exists a Pauli channel $\Lambda_A$ that acts on $\mathcal{H}_A$ such that $\Lambda_{AB}(\ket{\Phi^+}_{AB})=(\Lambda_{A}\otimes \mathcal{I}d)(\ket{\Phi^+}_{AB})$. Here $\mathcal{I}d$ denotes the identity channel. This can be shown by again considering Eq.~\eqref{eq:useful_identity}, For a Pauli operator in space $\mathcal{H}_A\otimes\mathcal{H}_B$, we can decompose it as $P\otimes Q$ where $P\in\mathcal{H}_A,Q\in\mathcal{H}_B$. Then $(P\otimes Q)\ket{\Phi^+}=(PQ^T\otimes I)\ket{\Phi^+}$, which becomes an error operator on system $\mathcal{H}_A$.

When an ideal EPR pair is subjected to Pauli noise, the system evolves into a Bell-diagonal state. Consider, for instance, the maximally entangled state $|\Phi^+\rangle = \frac{1}{\sqrt{2}}(|00\rangle+|11\rangle)$. If one subsystem undergoes a local depolarizing channel, the global state transforms into the density matrix:
\begin{equation}
    \rho = (1-p)|\Phi^+\rangle\langle\Phi^+| + \frac{p}{3}\left(|\Phi^-\rangle\langle\Phi^-| + |\Psi^+\rangle\langle\Psi^+| + |\Psi^-\rangle\langle\Psi^-|\right).
\end{equation}
This highly symmetric mixed state, representing a uniform mixture of the target Bell state with the orthogonal error states, is also called the Werner state.

For quantum measurement, there is a set of measurement basis of particular interest, which is called the Mutually Unbiased Basis (MUB). 
\begin{definition}[Mutually unbiased bases]
Let $\mathcal{H}$ be a $d$-dimensional Hilbert space. Two orthonormal bases
\[
\mathcal{B} = \{\,\lvert e_i \rangle \,\}_{i=1}^d,
\qquad
\mathcal{B}' = \{\,\lvert f_j \rangle \,\}_{j=1}^d
\]
are called \emph{mutually unbiased} if
\[
\bigl\lvert \langle e_i \mid f_j \rangle \bigr\rvert^2 = \frac{1}{d}
\quad
\text{for all } i,j .
\]
\end{definition}

Suppose in basis $\mathcal{B}$ we prepare a classical mixture of the basis vectors, i.e. 
\begin{equation}
    \rho_{\mathcal{B}}=\sum_{i}\lambda_i\ketbra{e_i}{e_i}.
\end{equation}
And we perform measurement in $\mathcal{B}'$ basis. The probability of outcome $\ket{f_j}$ occurs is
\begin{equation}
\begin{split}
    p_j&=\bra{f_j}\sum_{i}\lambda_i\ketbra{e_i}{e_i}f_j\rangle\\
    &=\frac{1}{d}\sum_{i}\lambda_i=\frac{1}{d}.
\end{split}
\end{equation}
That is to say, the measurement outcome in a mutually unbiased basis is uniformly random for any classical mixture state in one basis. 

The distance of two quantum state is often measured via fidelity, which, for two mixed state $\rho$ and $\sigma$, is
\begin{equation}
    F(\sigma,\rho)=\left(\tr\sqrt{\sqrt{\sigma}\rho\sqrt{\sigma}}\right)^2.
\end{equation}
When $\sigma$ is a pure state, say $\ket{\psi}$, the above formula could be reduced to $F(\ket{\psi},\rho)=\bra{\psi}\rho\ket{\psi}$. In our paper, we will mainly focus on the case where $\ket{\psi}=\ket{\Phi^+}$, the EPR state. And the goal is to do some operations on state $\rho$ to increase its fidelity with $\ket{\Phi^+}$.

\subsection{Entanglement distillation}\label{appendix:EDP}
More details on entanglement distillation could be found in~\cite{BDSW}. Here we just do a brief review of basic one-way and two-way protocols. We first introduce the one-way hashing method which is very standard in quantum key distribution. Through the discussion of one-way method we motivates why two-way detection based method are desired and then introduce basic two-way protocols.

\subsubsection{One-way hashing method.}
The one-way hashing method is an error-correction based method. Due to its equivalency to classical post-processing in the QKD setting, it has become the standard in QKD. Also since it does not discard qubit probabilistically, its yield is higher than most two-way method in low noise region. 

The hashing method involves the following steps:
\begin{enumerate}
    \item Parameter estimation: Alice and Bob first choose some noisy states and measure in $X$ or $Z$ basis. After measurement they compare their results to estimate the bit error rate $e_b$ and phase error rate $e_p$. 
    \item Bit error correction: Following parameter estimation, suppose $n$ entangled pairs remain with a bit error rate $e_b$. Alice and Bob agree on a binary hashing matrix $B^{e}$ of dimension $n \times nH(e_b)$, drawn from a universal hashing family. They compute the hashing results into the final $nH(e_b)$ pairs (the "check" qubits) by performing CNOT operations determined by the matrix entries: for every non-zero element $B^{e}_{i,j}=1$, a CNOT gate is applied between the $i$-th data qubit and the corresponding check qubit.  Subsequently, both parties measure these check qubits in the $Z$-basis. Alice transmits her outcome vector $\mathbf{s}_A$ to Bob, who combines it with his local result $\mathbf{s}_B$ to compute the error syndrome $\mathbf{s} = \mathbf{s}_A \oplus \mathbf{s}_B$. Finally, Bob decodes this syndrome to identify the error configuration and applies the necessary bit-flip corrections to the remaining $n-nH(e_b)$ data pairs.
    \item Phase error correction: Similar to phase error correction, they agree on a hashing matrix $B^p$ of dimension $n(1-H(e_b))\times nH(e_p)$ and compute the hashing results into the current last $nH(e_p)$ qubits. This time they perform CNOT in the Hadamard basis and measure the last $nH(e_p)$ qubits in $X$-basis measurement. Similarly Alice sends her measurement result to Bob. Bob do the decoding and perform phase-flip correction on remaining $n(1-H(e_b)-H(e_p))$ qubits.  
\end{enumerate}
After performing both bit-error correction and phase-error correction, the remaining
$n\bigl(1 - H(e_b) - H(e_p)\bigr)$ pairs converge to perfect EPR pairs in the asymptotic limit $n \to \infty$. The corresponding asymptotic yield of the one-way hashing protocol is therefore 
\begin{equation}
    Y=1-H(e_b)-H(e_p).
\end{equation}
When $e_b = e_p \approx 11\%$, this yield vanishes, indicating that one-way entanglement distillation fails beyond this error threshold. Consequently, there exist entangled states that are not distillable by one-way protocols alone. In such cases, it is necessary to first reduce the error rates below this threshold—using, for example, a two-way or pre-processing distillation step—before applying one-way hashing to obtain high-fidelity EPR pairs. In the following we introduce the most basic two-way protocol, the recurrence protocol.

\subsubsection{Recurrence method.}
The recurrence method relies on a parity check protocol where Alice and Bob compare their local measurement results to determine the quality of their shared entanglement. By verifying the consistency of their outcomes, they can selectively preserve EPR pairs with higher fidelity while discarding those that are likely erroneous. This process effectively suppresses errors, distilling a smaller set of high-purity pairs from a larger, noisy ensemble. 

The recurrence protocol consists of the following steps:
\begin{enumerate}
    \item Pauli twirling: Alice and Bob select two noisy pairs. They perform Pauli twirling to transform these states into Bell diagonal states $\rho_{AB}\otimes\sigma_{A'B'}$.
    \item Bilateral basis adjustment: Alice and Bob apply a single-qubit Clifford operation $C$ and its complex conjugate $C^*$, respectively, to each input pair. The bilateral operation $C\otimes C^*$ preserves $\ket{\Phi^+}$ and permutes the other three Bell states. For the identical Bell diagonal inputs considered below, they choose this permutation so that the smallest non-target Bell probability is assigned to $p_{10}$, i.e., $p_{10}\leq\min\{p_{01},p_{11}\}$ after the adjustment. This basis adjustment is repeated before every recurrence round.
    \item Coupling: They both perform CNOT operations targeted at the second qubit, effectively transforming their state to
    \begin{equation}
        \varrho=U_{\mathrm{BCNOT}}(\rho_{AB}\otimes \sigma_{A'B'})U_{\mathrm{BCNOT}}^{\dagger},\qquad U_{\mathrm{BCNOT}}=\mathrm{CNOT}_{AA'}\mathrm{CNOT}_{BB'}.
    \end{equation}
    \item Measurement: They measure their targeted qubit and compare their results. They keep the control qubit if their measurement outcomes are identical. Otherwise they discard the control qubit.
    \item Recurrence: The surviving pairs are then grouped together, and the bilateral basis adjustment, coupling, and measurement steps are repeated. For identical Bell diagonal inputs with initial target fidelity greater than $1/2$, the fidelity increases towards unity, as shown below.
\end{enumerate}
To analyze the performance of this protocol, suppose initially Alice and Bob share states,
\begin{equation}
    \varrho_{\mathrm{ini}}=(p_{00}\ketbra{\Phi^+}{\Phi^+}+p_{01}\ketbra{\Psi^+}{\Psi^+}+p_{10}\ketbra{\Phi^-}{\Phi^-}+p_{11}\ketbra{\Psi^-}{\Psi^-})^{\otimes 2}.
\end{equation}
After performing CNOT operations, the state can pass the consistency check if there exists even number of $X$-errors. The overall probability of passing the check is thus
\begin{equation}
\begin{split}
    p_{pass}&=p_{00}^2+p_{01}^2+p_{10}^2+p_{11}^2+2p_{00}p_{10}+2p_{01}p_{11}\\
    &=(p_{00}+p_{10})^2+(p_{01}+p_{11})^2.
\end{split}
\end{equation}
The probability of the remaining state being the four Bell states are respectively,
\begin{equation}\label{eq:boost}
\begin{split}
    p_{\Phi^+}&=p_{00}^2+p_{10}^2,\\
    p_{\Psi^+}&=p_{01}^2+p_{11}^2,\\
    p_{\Phi^-}&=2p_{00}p_{10},\\
    p_{\Psi^-}&=2p_{01}p_{11}.
\end{split}
\end{equation}
Eq.~\eqref{eq:boost} could be derived using two very useful identities $\mathrm{CNOT}_{A,B}X_A=X_AX_B\mathrm{CNOT}_{A,B}$ and $\mathrm{CNOT}_{A,B}Z_B=Z_AZ_B\mathrm{CNOT}_{A,B}$. Full table of outcome of different Bell diagonal state in the recurrence protocol could be found in~\cite{BDSW}. After normalization of the passing probability $p_{pass}$, the state after the recurrence protocol would become
\begin{equation}
    \rho_{\mathrm{final}}=\frac{p_{00}^2+p_{10}^2}{p_{pass}}\ketbra{\Phi^+}{\Phi^+}+\frac{p_{01}^2+p_{11}^2}{p_{pass}}\ketbra{\Psi^+}{\Psi^+}+\frac{2p_{00}p_{10}}{p_{pass}}\ketbra{\Phi^-}{\Phi^-}+\frac{2p_{01}p_{11}}{p_{pass}}\ketbra{\Psi^-}{\Psi^-}.
\end{equation}
The bilateral basis adjustment guarantees fidelity improvement whenever the initial target fidelity exceeds $1/2$. To see this, write
\[
    F=p_{00},\qquad z=p_{10}\leq\frac{1-F}{3}.
\]
The acceptance probability and updated target fidelity are
\[
    p_{pass}=(F+z)^2+(1-F-z)^2,\qquad F'=\frac{F^2+z^2}{p_{pass}}.
\]
Their difference satisfies
\[
\begin{split}
    F'-F&=\frac{(2F-1)\bigl[F(1-F)-2Fz-z^2\bigr]}{p_{pass}}\\
    &\geq\frac{(2F-1)(1-F)(4F-1)}{9p_{pass}}>0
\end{split}
\]
for $1/2<F<1$. Repeating the bilateral basis adjustment and recurrence step therefore increases the fidelity monotonically towards unity. More explicitly, since $p_{pass}\leq1$ and the fidelity remains above its initial value $F_0>1/2$, the constant
\[
    c=\frac{(2F_0-1)(4F_0-1)}{9}>0
\]
satisfies $1-F'\leq(1-c)(1-F)$ at every round. Thus, after $r$ rounds, $1-F_r\leq(1-c)^r(1-F_0)$.

\subsection{Twirling}\label{appendix:twirling}
We introduce the concept of twirling operation and basic properties of it in this part. More details regarding this could be found at~\cite{Mele2024introductiontohaar}. The twirling operation is defined over an ensemble, or put it simple, elements plus probability, of quantum unitary operations, $\mathfrak{S}$. We consider the following $k$-th order twirling operation over $\mathfrak{S}$:
\begin{equation}
\mathcal{T}^k_{\mathfrak{S}}(X) = \mathbb{E}_{U\sim\mathfrak{S}} U^{\otimes k}X U^{\dagger\otimes k}.
\end{equation}
where $U$ is drawn from the ensemble $\mathfrak{S}$ and $X$ is the operator being twirled. The expectation is taken over the random unitary operations from $\mathfrak{S}$.

The mostly commonly seen unitary ensemble for twirling is the Haar random unitaries on a Hilbert space. Intuitively, this is the uniform distribution of all unitaries over a Hilbert space. Suppose $X$ is an operator on the $k$-fold Hilbert of the Haar random ensemble, then its $k$-th order twirling satisfying the following equation:
\begin{equation}
\begin{split}
    \mathcal{T}_{\mathrm{Haar}}^{(k)}(X)&=\int_{U}U^{\otimes k}XU^{\dagger\otimes k}\mathrm{d}U\\
    &=\sum_{\pi,\sigma\in S_k}\mathrm{Wg}(\pi^{-1}\sigma, d)\tr(XS_{\sigma}^{\dagger})S_{\pi},
\end{split}
\end{equation}
where $S_{k}$ is the permutation group over dimension $k$. Operator $S_{\pi}$ is the permutation operator over permutation $\pi$. For example if $\pi$ is a length-2 cycle interchanging system $A$ and $B$, $S_{\pi}$ would be the SWAP operator on system $AB$, which we would also specifically termed as $F_{AB}$. $\textrm{Wg}(\sigma, d)$ denotes the Weingarten function for permutation $\sigma$ over dimension $d$. For large $d$, we have $\mathrm{Wg}(1, d)\approx\frac{1}{d^{k}}$, while other permutations decay as $O(d^{-n - | \pi |})$, where $| \pi |$ is the minimal number of transpositions needed to generate $\pi$.

For low order twirling, we can write the expressions in closed form. For example for the first order twirling, we would have:
\begin{equation}
    \mathcal{T}_{\mathrm{Haar}}^{(1)}(X)=\tr(X)\frac{\mathbb{I}_d}{d},
\end{equation}
where $\mathbb{I}_{d}$ is the identity operator with dimension $d$. The second-order twirling over Haar random ensemble is,
\begin{equation}\label{eq:secondtwirling}
\begin{split}
    \mathcal{T}_{\mathrm{Haar}}^{(2)}(X)=\tr(X \frac{\mathbb{I}_{d^2}-d^{-1} F}{d^2-1})\id_{d^2}+\tr(X \frac{F-d^{-1} \id_{d^2}}{d^2-1})F.
\end{split}
\end{equation}
Similarly $\id_{d^2}$ is the identity operator on the $d^2$-dimension subsystem, and $F$ denotes the SWAP operator on the $d^2$-dimension subsystem. This second order twirling is also called bilateral twirling. In our paper, another variant called the mixed bilateral twirling is also considered. The mixed bilateral twirling over a unitary ensemble $\mathfrak{S}$ is
\begin{equation}
    \mathcal{N}_{\mathfrak{S}}(X)=\mathbb{E}_{U\sim\mathfrak{S}}(U\otimes U^{*})X(U^{\dagger}\otimes U^{T}).
\end{equation}
The difference of this channel with the ordinary bilateral twirling channel is that it put a complex conjugate on unitaries on the second system. This mixed twirling channel is of particular interest in entanglement distillation because it naturally preserves the EPR pairs.

The Haar random unitary ensemble  requires exponential resources to implement in practice. An approximate version of the Haar random unitary ensemble for low order twirling is often used for practical applications. An ensemble $\mathfrak{S}$ is called a unitary $k$-design if $\mathcal{T}^{(k)}_{\mathfrak{S}} = \mathcal{T}^{(k)}_{\mathrm{Haar}}$. It was shown that the $n$-qubit Clifford group is an exact unitary 3-design~\cite{Webb2016Clifford3design,Zhu2017MultiqubitClifford}. As we can see in later proof, the most general Hamiltonian twirling requires $4$-design unitary ensemble, for which a Clifford group is not enough. 

Another very basic twirling operator we would use is the Pauli twirling, defined as
\begin{equation}
\begin{split}
    \mathcal{T}_{\mathsf{P}}^{(k)}(X)&=\sum_{P\in\mathsf{P}}P^{\otimes k}XP^{\otimes k},
\end{split}
\end{equation}
where $\mathsf{P}$ is the group of all Pauli operators in the Hilbert space. One good property of the Pauli channel is that the first order Pauli twirling of an arbitrary quantum channel would makes the channel into a Pauli channel. That is to say, channel $\Lambda_{\mathcal{E}}(\cdot)$, defined as
\begin{equation}
\begin{split}
    \Lambda_{\mathcal{E}}(\rho)&=\sum_{P\in\mathsf{P}}P\mathcal{E}(P\rho P)P
\end{split}
\end{equation}
would become a Pauli channel. Here we have abused the notion of twirling a bit by considering the twirling over a channel instead of a state. But the idea of these two kinds of twirling is similar. 

In this work, we will use the unitary evolution of Hamiltonian with random time as the unitary ensemble for twirling. For certain Hamiltonian $H$ and a time measure $\mu$, this ensemble could be constructed as $\mathfrak{S}_{H}=(e^{-iHt},\mu(t))$. And a major part of this paper is to study the property of mixed bilateral Hamiltonian twirling channel in the form of,
\begin{equation}
    \mathcal{N}_{H}(\rho)=\int_{t}(e^{-iHt}\otimes e^{iH^*t})\rho(e^{-iHt}\otimes e^{iH^*t})^{\dagger}\mu(t).
\end{equation}
For readers interested in more materials about the diagonal design induced by eigenvalues of the Hamiltonian, we refer to~\cite{HamiltonianShadow} for more details. 

\subsection{Out-of-Time-Order Correlator}
In this part, we briefly introduce the definition and basis properties of OTOC. 

\begin{definition}
    Consider a Hamiltonian $H$ on a Hilbert space $\mathcal{H}$. Let $\rho$ be a density matrix also in $\mathcal{H}$. Let $W$ and $V$ be two Hermitian  operators acting on specific subsystems in $\mathcal{H}$. The Out-of-Time-Order Correlator is defined as
    \begin{equation}
        \mathrm{OTOC}_{V,W}^{H}(t)=\tr(\rho \tilde{V}^{\dagger}W^{\dagger}\tilde{V}W)
    \end{equation}
    where $\tilde{V}=e^{-iHt}Ve^{iHt}$.
\end{definition}

In the above definition, $\rho$ is typically chosen as the thermal Gibbs state $\rho = e^{-\beta H} / \tr(e^{-\beta H})$, where $\beta = 1/(k_B T)$ is the inverse temperature. In our case, we would select $\rho=\mathbb{I}/d$ to be the identity. This can also be viewed as a Gibbs state at infinitely high temperature. 

OTOC could also be related to commutator, note that $[\tilde{V},W]=\tilde{V}W-W\tilde{V}$, suppose $U$ and $V$ are unitary operators,
\begin{equation}
\begin{split}
    \frac{1}{d}\tr([\tilde{V},W]^{\dagger}[\tilde{V},W])&=\frac{1}{d}\tr\Bigl((W^{\dagger}\tilde{V}^{\dagger}-\tilde{V}^{\dagger}W^{\dagger})(\tilde{V}W-W\tilde{V})\Bigr)\\
    &=\frac{1}{d}\tr\Bigl(W^{\dagger}\tilde{V}^{\dagger}\tilde{V}W-W^{\dagger}\tilde{V}^{\dagger}W\tilde{V}-\tilde{V}^{\dagger}W^{\dagger}\tilde{V}W+\tilde{V}^{\dagger}W^{\dagger}W\tilde{V}\Bigr)\\
    &=\frac{1}{d}\tr\Bigl(2I-W^{\dagger}\tilde{V}^{\dagger}W\tilde{V}-\tilde{V}^{\dagger}W^{\dagger}\tilde{V}W\Bigr)\\
    &=2-\frac{1}{d}\tr(\tilde{V}^{\dagger}W^{\dagger}\tilde{V}W)-\frac{1}{d}\tr\Bigl((\tilde{V}^{\dagger}W^{\dagger}\tilde{V}W)^{\dagger}\Bigr)\\
    &=2-(\mathrm{OTOC}_{V,W}^{H}+(\mathrm{OTOC}_{V,W}^{H})^{*}).
\end{split}
\end{equation}
The last line is due to the fact that $\tr(X)=\tr(X^{T})$. Therefore $\tr(X^{\dagger})=\tr(X^{*})=\tr(X)^{*}$. This connection shows that lower OTOC directly implies higher non-commutativity. 

In the context of quantum information, we often consider the OTOC for two local operators, $V$ and $W$, acting on spatially separated subsystems. Consider, for example, an $n$-qubit register where $V$ acts on the first qubit and $W$ on the last. At time $t=0$, these operators act on disjoint Hilbert spaces and therefore commute, yielding a OTOC value of $1$. However, as the system evolves under the Hamiltonian $H$, the operator $\tilde{V}(t) = e^{-iHt} V e^{iHt}$ undergoes operator spreading. The initially local operator $\tilde{V}(t)$ grows in support, evolving into a complex, non-local operator that eventually overlaps with the support of $W$. This emerging non-commutativity causes the OTOC value to decay. This phenomenon, known as information scrambling, quantifies how local quantum information is delocalized with the certain Hamiltonian dynamics. This plays a central role in the characterization of complex many-body dynamics as well as a wide range of quantum information tasks. 

\subsection{Gaussian unitary ensemble}\label{appendix:GUE}
To study the finite-time behavior of Hamiltonian distillation, we consider a standard ensemble of random Hamiltonians, the Gaussian unitary ensemble (GUE). Here we introduce its definition and the spectral properties used in Theorem~\ref{thm:O1}. Throughout this subsection, the Hilbert space dimension is $d=2^n$.

\begin{definition}[Gaussian unitary ensemble]
    A Hermitian matrix $H$ is drawn from the normalized $d\times d$ GUE if its entries satisfy
    \begin{equation}
        H_{aa}=\frac{x_{aa}}{\sqrt{d}},\qquad
        H_{ab}=\frac{x_{ab}+iy_{ab}}{\sqrt{2d}},\qquad
        H_{ba}=H_{ab}^{*}\quad (a<b),
    \end{equation}
    where all real random variables $x_{aa}$, $x_{ab}$ and $y_{ab}$ are independent Gaussian variables with zero mean and unit variance. In particular, $\mathbb{E}H_{ab}=0$ and $\mathbb{E}|H_{ab}|^2=1/d$.
\end{definition}

The probability density of this ensemble is proportional to $\exp[-d\tr(H^2)/2]$. Consequently, its distribution is invariant under $H\mapsto VHV^{\dagger}$ for any fixed unitary $V$. In the spectral decomposition $H=UDU^{\dagger}$, the eigenbasis can therefore be represented by a Haar random unitary $U$ independent of the eigenvalues in $D$. This property allows us to evaluate the ensemble average in Theorem~\ref{thm:Haar} by first fixing the eigenvalues and averaging over the Haar random eigenbasis. The eigenvalue distribution also determines the finite-time behavior studied in Theorem~\ref{thm:O1}.

With the above normalization, the empirical eigenvalue distribution converges as $d\to\infty$ to the semicircle law,
\begin{equation}
    \rho_{\mathrm{sc}}(\lambda)=\frac{\sqrt{4-\lambda^2}}{2\pi},\qquad -2\leq\lambda\leq 2,
\end{equation}
and is zero outside this interval. The characteristic energy scale thus remains of order one as the number of qubits increases. For any fixed time $t$, the normalized trace of the evolution converges to the Fourier transform of this density,
\begin{equation}
    \frac{1}{d}\tr(e^{-iHt})\longrightarrow
    \varphi(t)=\int_{-2}^{2}e^{-i\lambda t}\rho_{\mathrm{sc}}(\lambda)\mathrm{d}\lambda
    =\frac{J_1(2t)}{t},
\end{equation}
where $J_1$ is the Bessel function of the first kind of order $1$, and $\varphi(0)=1$ by continuity.

The GUE is a Gaussian example of the Hermitian Wigner matrices appearing in the OTOC result of~\cite{cipolloni2024outoftimeorderedcorrelatorswignermatrices}. For Pauli observables, that result expresses the leading OTOC in terms of $\varphi(t)$ and $\varphi(2t)$. At the first positive zero $\tau_1=j_{1,1}/2\approx 1.916$ of $\varphi(t)$, these leading contributions vanish. This is the spectral property used in the proof of Theorem~\ref{thm:O1} to obtain distillation at a time independent of $n$. The constant time is understood in the above energy normalization; rescaling the Hamiltonian by a positive factor rescales the required evolution time by its inverse.

\section{Distillation proofs and performance bounds}\label{appendix:proofs}
\subsection{Proof of Theorem~\ref{thm:OTOC}}\label{appendix:OTOC}

\begin{thmrecap}[Theorem~\ref{thm:OTOC}]
    For the Pauli noise and conjugate measurements specified in Sec.~\ref{sec:main}, Hamiltonian distillation gives
    \begin{equation}
    \begin{split}
        f_{\mathrm{out}}&\geq\frac{c_I}{c_I+\mathrm{OTOC}_{\Lambda,M}^{\mathcal{N}_H}},\\
        Y&=\bigl(c_I+\mathrm{OTOC}_{\Lambda,M}^{\mathcal{N}_H}\bigr)(1-m/n),
    \end{split}
    \end{equation}
    with the averaged OTOC defined in Eq.~\eqref{eq:averaged_OTOC}.
\end{thmrecap}

\begin{proof}
Write $\Phi_n=\ketbra{\Phi^+}{\Phi^+}^{\otimes n}$ for the normalized ideal state and $U(t)=e^{-iHt}$. An input Pauli error evolves as
\begin{equation}
    (U(t)\otimes U^*(t))(P\otimes I)\ket{\Phi^+}^{\otimes n}
    =(\tilde P(t)\otimes I)\ket{\Phi^+}^{\otimes n},\qquad
    \tilde P(t)=U(t)PU^{\dagger}(t).
\end{equation}
Using $\langle\Phi_n|A\otimes B|\Phi_n\rangle=\tr(AB^T)/d$, where $|\Phi_n\rangle=\ket{\Phi^+}^{\otimes n}$, the probability that an input error $P$ passes the check at time $t$ is
\begin{equation}
\begin{split}
    a(P,t)&=\sum_x\tr\!\left[(M_x\otimes M_x^*)
      (\tilde P(t)\otimes I)\Phi_n(\tilde P^{\dagger}(t)\otimes I)\right]\\
    &=\frac{1}{d}\sum_x\tr\!\left[\tilde P^{\dagger}(t)M_x\tilde P(t)M_x\right]
      =\sum_x\mathrm{OTOC}_{P,M_x}^{e^{-iHt}}.
\end{split}
\end{equation}
For $P=I$, Eq.~\eqref{eq:ideal_acceptance} gives $a(I,t)=1$, and the unmeasured pairs remain ideal. Thus
\begin{equation}
    p_{\mathrm{pass}}=c_I+\mathrm{OTOC}_{\Lambda,M}^{\mathcal{N}_H}.
\end{equation}
Let $R$ be the probability of rejection conditioned on an initial nonidentity error. For $c_I<1$,
\begin{equation}
    R=\frac{\int_t\sum_{P\neq I}c_P[1-a(P,t)]\mu(t)}{1-c_I}
     =\frac{1-c_I-\mathrm{OTOC}_{\Lambda,M}^{\mathcal{N}_H}}{1-c_I}.
\end{equation}
The accepted error-free component has weight $c_I$ and unit retained-block fidelity. Every other accepted component has nonnegative overlap with the ideal retained state, so
\begin{equation}\label{eq:rate&fidelity}
\begin{split}
    f_{\mathrm{out}}&\geq\frac{c_I}{c_I+(1-c_I)(1-R)},\\
    Y&=\bigl[c_I+(1-c_I)(1-R)\bigr](1-m/n).
\end{split}
\end{equation}
Substituting the expression for $R$ proves the theorem. If $c_I=1$, directly $f_{\mathrm{out}}=1$ and $Y=1-m/n$; no conditional detection rate is needed. Equality in the fidelity bound holds only when all initially erroneous accepted components have zero overlap with the ideal retained state. The additional numerator contribution is evaluated in Appendix~\ref{appendix:benign_errors}.
\end{proof}

\subsection{Proof of Corollary~\ref{cor:Z_OTOC}}\label{appendix:Z_OTOC}
\begin{proof}
The local computational-basis measurements have joint projectors $M_x=\ketbra{x}{x}\otimes I^{\otimes(n-m)}$, for $x\in\{0,1\}^m$, with the Pauli expansion
\begin{equation}
    M_x=\frac{1}{2^m}\sum_{s\in\{0,1\}^m}(-1)^{x\cdot s}Z^s.
\end{equation}
At fixed time $t$, let $\tilde{P}(t)=e^{-iHt}Pe^{iHt}$. The accepted nonidentity-error weight in Theorem~\ref{thm:OTOC} is
\begin{equation}
\begin{split}
    p_{\mathrm{err}}(t)
    &=\frac{1}{d}\sum_{P\in\mathsf{P}\backslash I}c_P\sum_x
      \tr\!\left[\tilde{P}^{\dagger}(t)M_x\tilde{P}(t)M_x\right]\\
    &=\frac{1}{d\,2^{2m}}\sum_{P\in\mathsf{P}\backslash I}c_P
      \sum_{x,s,r}(-1)^{x\cdot(s\oplus r)}
      \tr\!\left[\tilde{P}^{\dagger}(t)Z^s\tilde{P}(t)Z^r\right]\\
    &=\frac{1}{d\,2^m}\sum_{P\in\mathsf{P}\backslash I}c_P\sum_s
      \tr\!\left[\tilde{P}^{\dagger}(t)Z^s\tilde{P}(t)Z^s\right].
\end{split}
\end{equation}
The last equality uses $\sum_x(-1)^{x\cdot(s\oplus r)}=2^m\delta_{s,r}$. The $s=0$ term contributes $(1-c_I)2^{-m}$ because $\tilde{P}^{\dagger}(t)\tilde{P}(t)=I$, proving Eq.~\eqref{eq:Z_acceptance}. Ideal EPR pairs pass with probability one, so the total acceptance probability is $c_I+p_{\mathrm{err}}(t)$ and the yield follows. Appendix~\ref{appendix:benign_errors} shows that the overlap of the unnormalized accepted state with the ideal retained state is $c_I+b_{\Lambda,M}(t)$, proving the exact fidelity formula and its lower bound.
\end{proof}

\subsection{Exact retained-state fidelity and discarded errors}\label{appendix:benign_errors}
The event that an input block had no Pauli error is sufficient, but not necessary, for ideal retained entanglement. We make the distinction explicit here. Let $T$ denote the $m$ measured qubits, $R$ the $r=n-m$ retained qubits, and $d=2^n$. The measured qubits are discarded, and no outcome-dependent correction is applied to the retained qubits.

\paragraph{General conjugate measurements.}
Write the measurement effects on $T$ as $E_x$, so that the operators in Theorem~\ref{thm:OTOC} are $M_x=E_x\otimes I_R$. Let $\Phi_r$ be the normalized ideal state on the retained pairs and
\begin{equation}
    \rho_P(t)=(\tilde P(t)\otimes I)\Phi_n(\tilde P^{\dagger}(t)\otimes I),\qquad
    S_T=\sum_x E_x\otimes E_x^*.
\end{equation}
With tensor factors grouped into measured and retained pairs, the joint weight of acceptance and ideal retained-state overlap is
\begin{equation}\label{eq:general_good_weight}
\begin{split}
    g_P(t)&=\tr\!\left[(S_T\otimes\Phi_r)\rho_P(t)\right]\\
    &=\frac{1}{d\,2^r}\sum_x\tr\!\left[A_P^{\dagger}(t)E_x A_P(t)E_x\right],\qquad
      A_P(t)=\tr_R\tilde P(t).
\end{split}
\end{equation}
The second line follows by contracting the retained systems with $\bra{\Phi^+}^{\otimes r}$. This expression depends only on the measurement effects because the measured systems are discarded. Positivity gives $0\leq g_P(t)\leq a(P,t)$, and ideal acceptance gives $g_I(t)=1$. Consequently,
\begin{equation}\label{eq:exact_general_fidelity}
    f_{\mathrm{out}}=
    \frac{c_I+b_{\Lambda,M}^{\mathcal N_H}}{c_I+\mathrm{OTOC}_{\Lambda,M}^{\mathcal N_H}},\qquad
    b_{\Lambda,M}^{\mathcal N_H}=\int_t\sum_{P\neq I}c_Pg_P(t)\mu(t)\geq0.
\end{equation}
Thus the original error-free weight yields a lower bound, while $b_{\Lambda,M}^{\mathcal N_H}$ accounts for the additional overlap from initially erroneous blocks. A positive overlap does not require that each such block individually output a pure EPR state.

\paragraph{Local Pauli measurements.}
For computational-basis measurements, expand the dressed error in the Hermitian Pauli basis, with global phases omitted:
\begin{equation}
    \tilde P(t)=\sum_{Q\in\mathsf P}\alpha_{P,Q}(t)Q,\qquad
    \alpha_{P,Q}(t)=\frac{1}{d}\tr\!\left[Q\tilde P(t)\right].
\end{equation}
Matching outcomes require the measured part of a Pauli string to be $Z$-type. If $|x\rangle$ is the common outcome, the retained amplitude operator is $\langle x|\tilde P(t)|x\rangle$. Summing over $x$ and using
$\sum_x(-1)^{x\cdot(s\oplus s')}=2^m\delta_{s,s'}$ gives
\begin{equation}\label{eq:pauli_good_and_pass}
\begin{split}
    a(P,t)&=\sum_{s\in\{0,1\}^m}\sum_{Q_R\in\mathsf P_r}
        \left|\alpha_{P,Z_T^s\otimes Q_R}(t)\right|^2,\\
    g_P(t)&=\sum_{s\in\{0,1\}^m}
        \left|\alpha_{P,Z_T^s\otimes I_R}(t)\right|^2.
\end{split}
\end{equation}
Since $\tr\tilde P(t)=0$ for $P\neq I$, the $s=0$ term in $g_P$ vanishes for nonidentity errors. This proves Eqs.~\eqref{eq:benign_weight} and~\eqref{eq:local_exact_fidelity}. For Hadamard-basis measurements, replace $Z_T^s$ by $X_T^s$ throughout. The distinction is between all accepted Pauli components and the subset acting trivially on the retained qubits \emph{after evolution}; sums of initial measurement-only errors generally do not determine this subset.

For illustration, take $n=2$, $m=1$, any diagonal Hamiltonian, computational-basis measurement of pair 1, and the input $[c_I\Phi^++(1-c_I)\Phi^-]_1\otimes\Phi^+_2$. Both components pass, and the only nonidentity error is $Z_1\otimes I_2$. Thus $p_{\mathrm{pass}}=1$, $b=1-c_I$, and $f_{\mathrm{out}}=1$, as required.

When random evolution times or Hamiltonians are sampled before post-selection, the unnormalized accepted state must be averaged first. Equivalently, average $p_{\mathrm{pass}}$ and $p_{\mathrm{pass}}f_{\mathrm{out}}$ separately and then take their ratio. This is generally different from an unweighted average of the conditional fidelities.

\paragraph{Explicit diagonal-Hamiltonian formulas.}
Write $P=X^aZ^b$ up to an irrelevant global phase, with probabilities $c_{a,b}$ and bit masks $a=(a_T,a_R)$, $b=(b_T,b_R)$. For a computational-basis check and any diagonal Hamiltonian, at every evolution time,
\begin{equation}\label{eq:diagonal_Z_fidelity}
    p_{\mathrm{pass}}^{(Z)}=\sum_{a:a_T=0}\sum_b c_{a,b},\qquad
    f_{\mathrm{out}}^{(Z)}=
    \frac{\sum_{b_T}c_{0,(b_T,0_R)}}{p_{\mathrm{pass}}^{(Z)}}.
\end{equation}
Indeed, diagonal evolution preserves the bit-flip mask, so acceptance requires $a_T=0$. A component with $a_R\neq0$ has zero ideal retained-state overlap. If $a=0$, the phase error commutes with the Hamiltonian, and the retained state is ideal precisely when $b_R=0$.

For the Hadamard-basis check used in Sec.~\ref{sec:tolerance}, assume long-time twirling and the nondegenerate-gap conditions of Eq.~\eqref{eq:diagonal_error_appendix}. Define
\begin{equation}\label{eq:diagonal_noise_weights}
    C_Z=\sum_b c_{0,b},\qquad
    C_{Z,R}=\sum_{b:b_T=0}c_{0,b},\qquad
    C_{X,T}=\sum_{\substack{a\neq0\\a_R=0}}\sum_b c_{a,b}.
\end{equation}
Here $C_{X,T}$ is the total weight of errors with at least one bit flip, all bit flips confined to the measured subsystem, and arbitrary phase errors on either subsystem. The exact acceptance probability and retained-block fidelity are
\begin{equation}\label{eq:diagonal_X_fidelity}
    p_{\mathrm{pass}}^{(X)}=C_{Z,R}+2^{-m}(1-C_Z),\qquad
    f_{\mathrm{out}}^{(X)}=
    \frac{c_I+2^{-n}C_{X,T}}{C_{Z,R}+2^{-m}(1-C_Z)}.
\end{equation}
To verify the numerator, a pure phase-error component can be accepted with unit retained fidelity only when it is the identity. A component with a nonzero bit-flip mask becomes
\begin{equation}
    2^{-n}\sum_j|j\oplus a,j\rangle\langle j\oplus a,j|
\end{equation}
under the diagonal twirl. It passes the Hadamard check with probability $2^{-m}$, independently of $j$. Its conditional retained state has overlap $2^{-r}$ with $\Phi_r$ if $a_R=0$, and zero otherwise. Multiplying these factors gives the additional numerator $2^{-n}C_{X,T}$.

For local i.i.d.\ depolarizing noise,
\begin{equation}
    C_{X,T}=\left(1-\frac p2\right)^{n-m}
        \left[1-\left(1-\frac p2\right)^m\right],
\end{equation}
so, with $p_{\mathrm{pass}}$ from Eq.~\eqref{eq:QKD_parameters},
\begin{equation}\label{eq:diagonal_depolarizing_fidelity}
    f_{\mathrm{out}}^{(X)}=
    \frac{\left(1-\frac{3p}{4}\right)^n
      +2^{-n}\left(1-\frac p2\right)^{n-m}
       \left[1-\left(1-\frac p2\right)^m\right]}{p_{\mathrm{pass}}}.
\end{equation}
The bit- and phase-error rates in Eqs.~\eqref{eq:eb} and~\eqref{eq:ep} are obtained directly from the accepted state, rather than from $c_I/p_{\mathrm{pass}}$. They and the resulting QKD threshold are therefore unchanged by this fidelity correction.

\subsection{Proof of Theorem~\ref{thm:Haar} and exact GUE fidelity}\label{appendix:GUE_exact_fidelity}

For computational-basis tests, the correction to the fidelity numerator can also be evaluated exactly in the ensemble of Theorem~\ref{thm:Haar}. Let $T$ denote the $m$ measured qubits and $R$ the $n-m$ retained qubits, and define the phase-free Pauli sets
\begin{equation}
    \mathcal A=\{Z_T^s\otimes Q_R:s\in\{0,1\}^m,\ Q_R\in\mathsf P_{n-m}\},\qquad
    \mathcal B=\{Z_T^s\otimes I_R:s\in\{0,1\}^m\}.
\end{equation}
The set $\mathcal A$ contains all Pauli components that pass the measurement, while $\mathcal B\subseteq\mathcal A$ contains the components that also leave the retained pairs in the target EPR state. Write
\begin{equation}\label{eq:GUE_coefficient_sums}
    c_{\mathcal A}=\sum_{P\in\mathcal A\setminus\{I\}}c_P,\qquad
    c_{\mathcal B}=\sum_{P\in\mathcal B\setminus\{I\}}c_P,\qquad
    w=1-c_I,\quad \ell=2^m,\quad L=d^2/\ell.
\end{equation}
These are weights in the \emph{initial} Pauli channel. In particular, $c_{\mathcal B}$ collects nonidentity phase errors supported entirely on the measured qubits, whereas $c_{\mathcal A}$ allows arbitrary Pauli errors on the retained qubits.

\begin{proposition}[Exact long-time GUE correction]\label{prop:GUE_exact_fidelity}
Under the long-time ensemble averaging of Theorem~\ref{thm:Haar}, define
\begin{equation}\label{eq:GUE_transfer_constants}
    \beta_d=\frac{d+4}{d(d+1)(d+3)},\qquad
    \alpha_d=\frac{2(d+2)}{d(d+1)(d+3)},\qquad
    \gamma_d=\frac{2}{d^2(d+1)(d+3)}.
\end{equation}
The total accepted weight arising from nonidentity input errors and its contribution to the unnormalized retained-state fidelity are, respectively,
\begin{align}
    p_{\mathrm{err}}^{\mathrm{GUE}}
    &=\beta_d(L-1)w+\alpha_d c_{\mathcal A}
      +\gamma_d L(w-c_{\mathcal B}),\label{eq:GUE_exact_acceptance}\\
    b_{\mathrm{GUE}}
    &=\beta_d(\ell-1)w+\alpha_d c_{\mathcal B}
      +\gamma_d\ell(w-c_{\mathcal A}).\label{eq:GUE_exact_benign}
\end{align}
Consequently, when the unnormalized accepted state is averaged before normalization,
\begin{equation}\label{eq:GUE_exact_fidelity}
    f_{\mathrm{out}}
    =\frac{c_I+b_{\mathrm{GUE}}}{c_I+p_{\mathrm{err}}^{\mathrm{GUE}}},\qquad
    Y=(c_I+p_{\mathrm{err}}^{\mathrm{GUE}})(1-m/n).
\end{equation}
All three terms on the right-hand side of Eq.~\eqref{eq:GUE_exact_benign} are nonnegative. In particular,
\begin{equation}\label{eq:GUE_explicit_numerator_bound}
    b_{\mathrm{GUE}}\geq
    \beta_d(\ell-1)(1-c_I)+\alpha_d c_{\mathcal B}.
\end{equation}
For local i.i.d.\ depolarizing noise of strength $p$, these additional weights have the explicit forms
\begin{equation}\label{eq:GUE_depolarizing_sums}
    c_I=(1-3p/4)^n,\qquad
    c_{\mathcal A}=(1-p/2)^m-c_I,\qquad
    c_{\mathcal B}=(1-p/2)^m(1-3p/4)^{n-m}-c_I.
\end{equation}
\end{proposition}

\begin{proof}
Write $H=VDV^\dagger$, with $V$ Haar distributed and independent of the eigenvalues. The GUE eigenvalue distribution has a joint density, so eigenvalue degeneracies and nontrivial equalities between energy gaps occur with probability zero. We may therefore condition on any spectrum with distinct eigenvalues and nondegenerate nonzero gaps; the long-time result below is independent of the particular spectrum. For nonidentity phase-free Hermitian Paulis $P,Q$, introduce the long-time averaged transfer weight
\begin{equation}
    K(P,Q)=\mathbb E_V\,\overline{\left|
       d^{-1}\operatorname{tr}\!\left(Qe^{-iHt}Pe^{iHt}\right)
       \right|^2}^{\,t}.
\end{equation}
The sign convention for time does not affect this long-time average. With $P'=V^\dagger PV$ and $Q'=V^\dagger QV$, the non-degenerate-gap assumption gives
\begin{equation}\label{eq:GUE_transfer_time_average}
    K(P,Q)=\frac{1}{d^2}\mathbb E_V\left[
      \left(\sum_i P'_{ii}Q'_{ii}\right)^2
      +\sum_{i\ne j}|P'_{ij}|^2|Q'_{ij}|^2\right].
\end{equation}
We evaluate this expression using moments of one or two Haar-distributed orthonormal vectors. Set
\begin{equation}
    \delta=\delta_{P,Q}=d^{-1}\operatorname{tr}(PQ),\qquad
    \chi=d^{-1}\operatorname{tr}(PQPQ)\in\{1,-1\}.
\end{equation}
Let $u,v$ be two distinct columns of $V$. Define $p_u=u^\dagger Pu$, $q_u=u^\dagger Qu$, and $h_u=u^\dagger(PQ+QP)u$, and the moments
\begin{equation}
    A_4=\mathbb E[p_u^2q_u^2],\quad
    B_4=\mathbb E[p_uq_u p_vq_v],\quad
    C_4=\mathbb E[|u^\dagger Pv|^2|u^\dagger Qv|^2].
\end{equation}
For a Haar unit vector,
\begin{equation}
    \mathbb E\bigl[(uu^\dagger)^{\otimes k}\bigr]
    =\frac{\sum_{\pi\in S_k}W_\pi}{d(d+1)\cdots(d+k-1)},
\end{equation}
where $W_\pi$ permutes the $k$ tensor factors. Applying this identity with $k=2,3,4$, using $P^2=Q^2=I$ and $\operatorname{tr}P=\operatorname{tr}Q=0$, gives
\begin{align}
    A_4&=\frac{d(1+2\delta)+4+2\chi}{(d+1)(d+2)(d+3)},\label{eq:GUE_sphere_fourth}\\
    \mathbb E[p_uq_u]&=\frac{\delta}{d+1},\qquad
    \mathbb E[p_u^2]=\mathbb E[q_u^2]=\frac{1}{d+1},\\
    \mathbb E[|u^\dagger PQu|^2]&=\frac{d\delta+1}{d+1},\qquad
    \mathbb E[p_uq_u h_u]=\frac{2(d\delta+\chi+1)}{(d+1)(d+2)}.
\end{align}
Conditionally on $u$, the vector $v$ is Haar distributed in the $(d-1)$-dimensional subspace with projector $R_u=I-uu^\dagger$. Its second moment is
\begin{equation}
    \mathbb E_v\bigl[(vv^\dagger)^{\otimes2}\mid u\bigr]
    =\frac{(R_u\otimes R_u)(I+W_{(12)})}{d(d-1)}.
\end{equation}
For $B_4$, the conditional expectation of $p_vq_v$ is
$[\operatorname{tr}(PQ)-h_u+2p_uq_u]/[d(d-1)]$.
For $C_4$, the corresponding numerator is
$(1-p_u^2)(1-q_u^2)+|u^\dagger PQu-p_uq_u|^2$.
Thus
\begin{align}
    d(d-1)B_4
    &=\frac{d\delta}{d+1}
      -\frac{2(d\delta+\chi+1)}{(d+1)(d+2)}+2A_4,\\
    d(d-1)C_4
    &=1-\frac{2}{d+1}+\frac{d\delta+1}{d+1}
      -\frac{2(d\delta+\chi+1)}{(d+1)(d+2)}+2A_4.
\end{align}
Substitution into Eq.~\eqref{eq:GUE_transfer_time_average} yields
\begin{align}
    K(P,Q)
    &=\frac{dA_4+d(d-1)(B_4+C_4)}{d^2}\\
    &=\frac{d^2+4d+2-2\chi+2d(d+2)\delta}{d^2(d+1)(d+3)}
      =\beta_d+\alpha_d\delta_{P,Q}+\gamma_d(1-\chi_{P,Q}).\label{eq:GUE_exact_transfer}
\end{align}

It remains to sum these transfer weights over accepted Pauli components. The Pauli commutation characters satisfy
\begin{equation}
    \sum_{Q\in\mathcal A}\chi_{P,Q}
       =L\,\mathbf 1_{P\in\mathcal B},\qquad
    \sum_{Q\in\mathcal B}\chi_{P,Q}
       =\ell\,\mathbf 1_{P\in\mathcal A}.
\end{equation}
These identities follow because $\mathcal A$ and $\mathcal B$ are each other's symplectic orthogonal complements. Since conjugation preserves the trace, a nonidentity input Pauli has zero identity-output coefficient. Hence
\begin{align}
    p_{\mathrm{err}}^{\mathrm{GUE}}
      &=\sum_{P\ne I}c_P\sum_{Q\in\mathcal A\setminus\{I\}}K(P,Q),\\
    b_{\mathrm{GUE}}
      &=\sum_{P\ne I}c_P\sum_{Q\in\mathcal B\setminus\{I\}}K(P,Q).
\end{align}
Combining these expressions with Eq.~\eqref{eq:GUE_exact_transfer} proves Eqs.~\eqref{eq:GUE_exact_acceptance}--\eqref{eq:GUE_exact_benign}. The identity input always passes and contributes $c_I$ to both the accepted weight and the fidelity numerator, establishing Eq.~\eqref{eq:GUE_exact_fidelity}. Finally, for depolarizing noise each measured factor in $\mathcal A$ or $\mathcal B$ has total $I/Z$ probability $1-p/2$; the retained factors in $\mathcal B$ must all be identity, each with probability $1-3p/4$. This proves Eq.~\eqref{eq:GUE_depolarizing_sums}.
\end{proof}

The exact acceptance formula also implies the uniform refinement
\begin{equation}
    p_{\mathrm{err}}^{\mathrm{GUE}}
    =\frac{1-c_I}{2^m}+O\!\left(\frac{1-c_I}{d^2}\right).
\end{equation}
Optimizing the exact formulas over the initial noise weights gives the sharp bound in Theorem~\ref{thm:Haar}, as shown below. The extra numerator is of order at most $2^m(1-c_I)/d^2$, but it should not be discarded as an exact identity or assumed negligible relative to $c_I$ without checking the parameter regime. These formulas concern long-time GUE evolution; replacing it with a Haar-distributed evolution unitary would give different finite-$d$ transfer weights.

\paragraph{Sharp finite-size bounds.}
The exact GUE formulas also give a sharp distribution-independent fidelity bound, avoiding a signed asymptotic remainder in the denominator. Write $c=c_I$, $w=1-c$, $D_d=(d+1)(d+3)$, $\ell=2^m$, and $L=d^2/\ell$, and define
\begin{equation}\label{eq:GUE_explicit_probabilities}
\begin{split}
    B_{d,m}&=\beta_d(\ell-1)=\frac{(\ell-1)(d+4)}{dD_d},\\
    q^+_{d,m}&=\frac1\ell+\frac{1-\ell^{-1}}{D_d},\qquad
    q^-_{d,m}=\frac1\ell-\frac{1+\ell^{-1}+4/d}{D_d}.
\end{split}
\end{equation}
Using $0\leq c_{\mathcal B}\leq c_{\mathcal A}\leq w$, Eqs.~\eqref{eq:GUE_exact_acceptance} and~\eqref{eq:GUE_exact_benign} imply
\begin{equation}\label{eq:GUE_sharp_bound_appendix}
    f_{\mathrm{out}}\geq\frac{c+wB_{d,m}}{c+wq^+_{d,m}}.
\end{equation}
The acceptance probability and yield obey
\begin{equation}\label{eq:GUE_acceptance_bounds}
    c+wq^-_{d,m}\leq p_{\mathrm{pass}}\leq c+wq^+_{d,m},\qquad
    Y=p_{\mathrm{pass}}(1-m/n).
\end{equation}
Indeed, the fidelity numerator is minimized and the acceptance denominator is maximized simultaneously at $c_{\mathcal A}=w$, $c_{\mathcal B}=0$. At this point
\begin{equation}
    \beta_d(L-1)+\alpha_d+\gamma_dL=q^+_{d,m}.
\end{equation}
Equality in the fidelity bound is attained by any Pauli channel whose nonidentity errors lie in $\mathcal A\setminus\mathcal B$, for example an identity/Pauli mixture with nonidentity error $Z^{\otimes n}$, which saturates the bound for every choice of measured subset. The acceptance lower bound is attained when all nonidentity errors lie outside $\mathcal A$, since $\alpha_d>\gamma_dL$ for $m\geq1$. In particular, $q^+_{d,m}$ is an exact worst-case acceptance probability conditioned on an initial error; it is not the actual acceptance probability for every channel.

For a specified Pauli channel, the exact denominator may equivalently be written
\begin{equation}\label{eq:GUE_explicit_denominator}
    p_{\mathrm{pass}}
      =c+wq^+_{d,m}
       -\alpha_d(w-c_{\mathcal A})
       -\frac{2}{\ell D_d}c_{\mathcal B}.
\end{equation}
For a target fidelity $f\in(0,1)$, the sharp universal bound guarantees $f_{\mathrm{out}}\geq f$ precisely when
\begin{equation}\label{eq:GUE_target_condition}
    \frac{c}{1-c}\geq\frac{f q^+_{d,m}-B_{d,m}}{1-f},
\end{equation}
if the right-hand side is positive; otherwise the guarantee is automatic. The case $c=1$ is trivial. This is a necessary and sufficient condition for a guarantee uniform over the initial Pauli distribution at fixed $c$, rather than a necessary condition for each individual distribution. Equations~\eqref{eq:GUE_sharp_bound_appendix} and~\eqref{eq:GUE_acceptance_bounds}, together with their saturation, establish Theorem~\ref{thm:Haar}.

\subsection{Proof of Corollary~\ref{cor:GUE_threshold} and threshold regimes}\label{appendix:GUE_threshold}
\begin{proof}
Write $c=c_I$, $w=1-c$, $d=2^n$, $\ell=2^m$, and $D_d=(d+1)(d+3)$. Introduce the relative accepted identity weight and the retained-state overlap scale
\begin{equation}
    x_n=\frac{2^m c}{1-c},\qquad q_n=4^{-(n-m)}.
\end{equation}
The exact coefficients give, uniformly over $c_{\mathcal A},c_{\mathcal B}$,
\begin{equation}
    p_{\mathrm{err}}^{\mathrm{GUE}}=\frac{w}{\ell}+O(w/d^2),\qquad
    b_{\mathrm{GUE}}=\frac{w\ell}{d^2}+O(w/d^2).
\end{equation}
Multiplying the exact fidelity numerator and denominator by $\ell/w$ therefore yields
\begin{equation}\label{eq:GUE_uniform_ratio}
    f_{\mathrm{out}}=
      \frac{x_n+q_n}{1+x_n}
      +O\!\left(\frac{2^{m-2n}}{1+x_n}\right).
\end{equation}
The constants are independent of the initial Pauli distribution and of $c$. For example, $d^2\beta_d=1-3/D_d$, $\alpha_d=O(d^{-2})$, and $\gamma_d=O(d^{-4})$ directly establish the two expansions. Because $m<n$, the denominator's relative correction is $O(2^{m-2n})=O(2^{-n})$ even when $c$ is arbitrarily small. Under the assumptions of Corollary~\ref{cor:GUE_threshold},
\begin{equation}
    x_n\sim\kappa2^{(\rho-\lambda)n},\qquad
    q_n=2^{-2(1-\rho)n}\longrightarrow0.
\end{equation}
Substituting these limits into Eq.~\eqref{eq:GUE_uniform_ratio} gives the three cases in Eq.~\eqref{eq:GUE_threshold_corollary}. We take $m=\rho n$ along admissible integer block sizes; bounded measurement offsets are discussed below. Finally, $c\to0$ and $p_{\mathrm{err}}^{\mathrm{GUE}}\sim2^{-\rho n}$ imply
\begin{equation}
    Y\sim(1-\rho)\left(\kappa2^{-\lambda n}+2^{-\rho n}\right),
\end{equation}
which reduces to $(1-\rho)c_I$ when $\lambda<\rho$.
\end{proof}

\paragraph{Critical prefactors and target fidelity.}

Suppose $m=\rho n+s$ with $0<\rho<1$ and fixed offset $s$, along an integer sequence, and $c\sim\kappa2^{-\lambda n}$ with $\kappa,\lambda>0$. Then
\begin{equation}\label{eq:GUE_exponential_threshold}
    f_{\mathrm{out}}\longrightarrow
    \begin{cases}
        1,&\lambda<\rho,\\
        \displaystyle\frac{\kappa2^s}{1+\kappa2^s},&\lambda=\rho,\\
        0,&\lambda>\rho.
    \end{cases}
\end{equation}
For integer rounding without a fixed offset, the critical limit is determined by the actual limit of $x_n$; the exponent alone does not determine it. In the high-fidelity regime,
\begin{equation}
    1-f_{\mathrm{out}}\sim\frac{1}{\kappa2^s}\,2^{-(\rho-\lambda)n}.
\end{equation}
This follows by applying the exact expression to $1-f_{\mathrm{out}}=(p_{\mathrm{err}}-b_{\mathrm{GUE}})/(c+p_{\mathrm{err}})$, since $b_{\mathrm{GUE}}/p_{\mathrm{err}}\to0$ when $n-m\to\infty$. At the critical exponent a fixed target $f$ requires the prefactor $\kappa2^s\geq f/(1-f)$ in the limiting expression, with a strict margin giving an eventual guarantee. Equivalently, for a growing retained block the useful measurement threshold is
\begin{equation}
    m\gtrsim\log_2\frac{1-c_I}{c_I}+\log_2\frac{f}{1-f},
\end{equation}
with the exact finite-size criterion given by Eq.~\eqref{eq:GUE_target_condition}.

For $s=0$ and $\lambda>\rho$, Eq.~\eqref{eq:GUE_uniform_ratio} further gives
\begin{equation}\label{eq:GUE_small_fidelity_scaling}
    f_{\mathrm{out}}\sim
      \kappa2^{-(\lambda-\rho)n}+2^{-2(1-\rho)n}.
\end{equation}
The extra numerator dominates only when $\lambda>2-\rho$, and the two terms are comparable at $\lambda=2-\rho$. Since $2-\rho>\rho$, this correction does not shift the threshold for constant or high fidelity when a positive fraction of pairs is retained.

\paragraph{Fixed number of retained pairs.}
If $m=n-r$ with fixed $r\geq1$, then $q_n=4^{-r}$ does not vanish, and instead
\begin{equation}\label{eq:GUE_fixed_retained}
    f_{\mathrm{out}}\longrightarrow
    \begin{cases}
        1,&\lambda<1,\\
        \displaystyle\frac{\kappa2^{-r}+4^{-r}}{1+\kappa2^{-r}},&\lambda=1,\\
        4^{-r},&\lambda>1.
    \end{cases}
\end{equation}
The last value is the overlap of a maximally mixed retained block with the target EPR state; a nonzero constant at fixed $r$ need not represent useful entanglement. At $c_I\sim\kappa2^{-n}$ and $r=1$, the limit is $(2\kappa+1)/(2\kappa+4)$, so a limiting fidelity of $0.9$ requires $\kappa\geq13$. The extra numerator therefore matters at this boundary. Fixed $\kappa$ gives a limit below one.

\paragraph{Depolarizing noise and yield.}
For local i.i.d.\ depolarizing noise, $c_I=(1-3p/4)^n$, hence $\lambda=-\log_2(1-3p/4)$. For fixed $0<\rho<1$, the high-fidelity condition becomes
\begin{equation}\label{eq:GUE_depolarizing_threshold}
    p<\frac43\left(1-2^{-\rho}\right).
\end{equation}
At equality with no measurement offset, $f_{\mathrm{out}}\to1/2$. As $\rho$ approaches one, the threshold approaches $p<2/3$; at $p=2/3$ the fixed-$r$ limit in Eq.~\eqref{eq:GUE_fixed_retained} is $2^{-r}$. These are statements about block fidelity for the GUE-averaged protocol; they do not replace the diagonal-Hamiltonian QKD analysis in Sec.~\ref{sec:tolerance}.

High fidelity at exponentially small $c_I$ has an exponentially small yield. When $m/n\to\rho<1$ and $\lambda<\rho$,
\begin{equation}
    Y\sim(1-\rho)\kappa2^{-\lambda n}.
\end{equation}
For fixed $r$ and $\lambda<1$, instead $Y\sim(r/n)\kappa2^{-\lambda n}$. Thus the high-fidelity limit does not assert a nonzero limiting yield per input pair as the block size diverges.

\subsection{Proof of Theorem~\ref{thm:O1}}\label{appendix:O1}
Corollary~\ref{cor:Z_OTOC} gives the accepted error weight at a fixed evolution time $t$. In the notation of Theorem~\ref{thm:OTOC}, this is
\begin{equation}
    \mathrm{OTOC}_{\Lambda,M}=p_{\mathrm{err}}(t)
    =\frac{1-c_I}{2^m}
    +\frac{1}{2^m}\sum_{P\in\mathsf{P}\backslash I}c_P
    \sum_{s\neq 0}\frac{1}{d}\tr\!\left(\tilde{P}^{\dagger}Z^s\tilde{P}Z^s\right),
\end{equation}
where $\tilde{P}=e^{-iHt}Pe^{iHt}$.
Our goal is therefore to prove that, at a finite time independent of $n$, the following bound holds with overwhelming probability over the GUE Hamiltonian for every nonidentity Pauli operator $P$ and every $s\neq 0$:
\begin{equation}\label{eq:O1target}
    \frac{1}{d}\tr(\tilde{P}^{\dagger}Z^{s}\tilde{P}Z^{s})=O_{\prec}(d^{-1}).
\end{equation}
Using Theorem~\ref{thm:OTOC}, this implies
\begin{equation}
\begin{split}
    f_{\mathrm{out}}&\geq\frac{c_{I}}{c_{I}+\mathrm{OTOC}_{\Lambda,M}}\\
    &=\frac{c_{I}}{c_{I}+(1-c_I)\bigl(2^{-m}+O_{\prec}(d^{-1})\bigr)},\\
    Y&=\Bigl[c_I+(1-c_I)\bigl(2^{-m}+O_{\prec}(d^{-1})\bigr)\Bigr](1-m/n).
\end{split}
\end{equation}

In the following we prove the equation in Eq.~\eqref{eq:O1target}. We will use the following lemma:

\begin{lemma}[OTOC for Wigner matrices\cite{cipolloni2024outoftimeorderedcorrelatorswignermatrices}]
    Let $W$ be a $d\times d$ Hermitian Wigner matrix with entry distributions $w_{aa}=\frac{1}{\sqrt{d}}\chi_{d}$ and $w_{ab}=\frac{1}{\sqrt{d}}\chi_{od}$, where $\chi_d$ and $\chi_{od}$ are diagonal and off-diagonal random variables satisfying
    \begin{itemize}
        \item $\chi_{od}$ is a complex centered variable satisfying $\mathbb{E}\chi_{od}=0,\mathbb{E}|\chi_{od}|^2=1$.
        \item $\chi_{d}$ is a real centered variable satisfying $\mathbb{E}\chi_{d}=0$.
        \item For any $p\in\mathbb{N}$, there exists $C_{p}>0$ such that $\mathbb{E}[|\chi_{d}|^{p}+|\chi_{od}|^{p}]\leq C_{p}$.
    \end{itemize}
    Then, for any deterministic traceless Hermitian operators $A,B$ satisfying $A^2=B^2=\mathbb{I}$ and any fixed $\varepsilon>0$, the OTOC satisfies
    \begin{equation}\label{eq:Wigner_involution_OTOC}
    \begin{split}
        \mathrm{OTOC}_{A,B}^{W,t}&=1-\frac{1}{d^2}\tr(A^2)\tr(B^2)[1-\varphi^2(t)]-\frac{2}{d^2}\tr(AB)^2\varphi^2(t)[\varphi(2t)-\varphi^2(t)]\\
        &\quad-\frac{1}{d}\tr(A^2B^2)\varphi^2(t)+\frac{1}{d}\tr(ABAB)\varphi(t)^4+O_{\prec}(\mathcal{E})
    \end{split}
    \end{equation}
    with the error term
    \begin{equation}\label{eq:Wigner_involution_error}
        \mathcal{E}=e^{t/d^{1/2-\varepsilon}}\left(\frac{|t|^4\tr(A^2)^2}{d^3}+\frac{|t|\tr(A^{8})^{1/2}}{d^{3/2}}\right)^{1/2}\left(\frac{|t|^4\tr(B^2)^2}{d^3}+\frac{|t|\tr(B^{8})^{1/2}}{d^{3/2}}\right)^{1/2}
    \end{equation}
    and $\varphi(t)=J_1(2t)/t$ where $J_1$ is the Bessel function of the first kind of order $1$.
\end{lemma}
Reference~\cite{cipolloni2024outoftimeorderedcorrelatorswignermatrices} formulates its result in terms of the commutator correlator
\[
    C_{A,B}(t)=\frac{1}{2d}\tr\!\left([A(t),B]^{\dagger}[A(t),B]\right),\qquad A(t)=e^{-iWt}Ae^{iWt}.
\]
Expanding the commutator and using cyclicity of the trace gives, for Hermitian operators,
\[
    C_{A,B}(t)=\frac{1}{d}\tr\!\left(A(t)^2B^2\right)-\frac{1}{d}\tr\!\left(A(t)BA(t)B\right).
\]
Under the assumptions $A^2=B^2=\mathbb{I}$, the first term equals one. Consequently, our four-point correlator satisfies
\[
    \mathrm{OTOC}_{A,B}^{W,t}=1-C_{A,B}(t),
\]
and Eqs.~\eqref{eq:Wigner_involution_OTOC}--\eqref{eq:Wigner_involution_error} follow from Theorem~2.2 of Ref.~\cite{cipolloni2024outoftimeorderedcorrelatorswignermatrices}.
We apply the lemma with $W=H$, $A=P$, and $B=Z^{s}$. Since $P\neq\mathbb{I}$ and $s\neq0$ in Eq.~\eqref{eq:O1target}, these are traceless Hermitian Pauli operators satisfying $P^2=(Z^{s})^2=\mathbb{I}$. Substituting into Eq.~\eqref{eq:Wigner_involution_OTOC} gives
\begin{equation}
\begin{split}
    \mathrm{OTOC}_{A,B}^{W,t}&=1-(1-\varphi^2(t))-2\delta_{A,B}\varphi^2(t)[\varphi(2t)-\varphi^2(t)]-\varphi^2(t)+\chi_{A,B}\varphi^4(t)+O_{\prec}(\mathcal{E})\\
    &=\chi_{A,B}\varphi^4(t)-2\delta_{A,B}\varphi^2(t)[\varphi(2t)-\varphi^2(t)]+O_{\prec}(\mathcal{E})
\end{split}
\end{equation}
The function $\varphi(t)$ has infinitely many zeros. Let $\tau_1=j_{1,1}/2\approx 1.916$ be its first positive zero, which is independent of $n$. At this time,
\begin{equation}
\begin{split}
    \mathrm{OTOC}_{A,B}^{W,\tau_1}=O_{\prec}(\mathcal{E})
\end{split}
\end{equation}
For Pauli observables and constant $\tau_1$, the error satisfies $\mathcal{E}=O(d^{-1})$, and hence Eq.~\eqref{eq:O1target} follows. Substitution into Theorem~\ref{thm:OTOC} gives both the fidelity and yield stated in Theorem~\ref{thm:O1}.

\subsection{Proof of Theorem~\ref{lemma:diagonal+classical}}\label{appendix:QKD_tolerance}
\begin{thmrecap}[Theorem~\ref{lemma:diagonal+classical}]
    For the diagonal-Hamiltonian QKD protocol of Sec.~\ref{sec:tolerance}, in the asymptotic regime of exponentially suppressed bit errors,
    \begin{equation}
        R=p_{\mathrm{pass}}\bigl(1-H(e_p)\bigr)(1-m/n),
    \end{equation}
    with $p_{\mathrm{pass}}$ and $e_p$ given in Eq.~\eqref{eq:QKD_parameters}. The tolerable depolarizing strength approaches the threshold $p_{\mathrm{tol}}\in[0,2/3]$ determined by
    \begin{equation}\label{eq:tolerance_appendix}
        \left(1-\frac{p_{\mathrm{tol}}}{2}\right)^{1-m/n}
        \left(2-\frac{3p_{\mathrm{tol}}}{2}\right)^{m/n}=1.
    \end{equation}
\end{thmrecap}

\begin{proof}
We first recall the action of a diagonal Hamiltonian $H_d=\mathrm{diag}(\lambda_1,\ldots,\lambda_d)$ on Pauli errors. Assume distinct eigenvalues and non-degenerate nonzero energy gaps. The long-time average retains only matrix elements whose accumulated phase vanishes, giving
\begin{equation}\label{eq:projection}
    \mathcal{N}_{H_d}(\rho)=\sum_{i\neq j}\left(\rho_{ii,jj}\ketbra{ii}{jj}+\rho_{ij,ij}\ketbra{ij}{ij}\right)
    +\sum_i\rho_{ii,ii}\ketbra{ii}{ii}.
\end{equation}
To see its effect on a Pauli error, write $P=X_0Z_0$ up to an overall phase, where $X_0$ and $Z_0$ are tensor products of $X$ and $Z$ operators, respectively, and set $\Phi_n=\ketbra{\Phi^+}{\Phi^+}^{\otimes n}$. If $X_0=I$, the error commutes with $H_d$ and the state is unchanged. Otherwise, the bit flips take every $\ket{ii}$ to a basis state with unequal Alice and Bob indices, so Eq.~\eqref{eq:projection} removes all coherences. Therefore,
\begin{equation}\label{eq:diagonal_error_appendix}
    \mathcal{N}_{H_d}\!\left[(P\otimes I)\Phi_n(P^{\dagger}\otimes I)\right]
    =\begin{cases}
        (Z_0\otimes I)\Phi_n(Z_0\otimes I),&X_0=I,\\
        (X_0\otimes I)\Delta(\Phi_n)(X_0\otimes I),&X_0\neq I,
    \end{cases}
\end{equation}
where $\Delta$ is complete dephasing on Alice's subsystem. Measuring $m$ pairs in the Hadamard basis thus accepts any component containing a bit flip with probability $2^{-m}$. For a component containing only phase errors, the check succeeds exactly when none of the measured pairs carries a phase error.

    These observations give three contributions to the accepted groups, with probabilities listed below:
    \begin{enumerate}
        \item Pure EPR states: The initial probability of pure EPR state is $(1-\frac{3}{4}p)^{n}$. They pass the error detection with $100\%$ probability, therefore the surviving probability of EPR states are also $p_1=(1-\frac{3}{4}p)^{n}$.
        \item EPR states affected by $Z$ errors: This kind of state can survive the error detection if $Z$-errors happen in the $n-m$ qubits while the $m$ qubits being measured is of pure EPR states. Therefore the probability of these kinds of state left is $p_2=(1-\frac{3p}{4})^{m}\left((1-\frac{p}{2})^{n-m}-(1-\frac{3p}{4})^{n-m}\right)$.
        \item EPR states also affected by $X$ errors: This kind of state pass the error detection with probability only $2^{-m}$. Therefore the overall surviving probability of them after error detection is $p_3=2^{-m}\left(1-(1-\frac{p}{2})^{n}\right)$.
    \end{enumerate}
    Now we calculate the bit and phase error rate of the surviving $n-m$ qubits after error detection. Denote $\overline{e}_b$ the average bit error rate over states that are affected by $X$-errors before distillation. Since overall bit error rate is $p/2$ before distillation for local depolarizing noise, we should have
    \begin{equation}
        \frac{\overline{e}_b\left(1-(1-\frac{p}{2})^{n}\right)+0\cdot(1-\frac{p}{2})^{n}}{\left(1-(1-\frac{p}{2})^{n}\right)+(1-\frac{p}{2})^{n}}=p/2~\Rightarrow~\overline{e}_b=\frac{p}{2\left(1-(1-\frac{p}{2})^{n}\right)}.
    \end{equation}
    After suppression of $X$ error, the remaining bit error rates are
    \begin{equation}\label{eq:eb}
    \begin{split}
        e_b&=\frac{2^{-m}\overline{e}_b\left(1-(1-\frac{p}{2})^{n}\right)}{2^{-m}\left(1-(1-\frac{p}{2})^{n}\right)+(1-\frac{3p}{4})^{m}(1-\frac{p}{2})^{n-m}}\\
        &=\frac{1}{1-(1-\frac{p}{2})^{n}+2^{m}(1-\frac{3p}{4})^{m}(1-\frac{p}{2})^{n-m}}\cdot\frac{p}{2}\\
        &=\frac{1}{1+(1-\frac{p}{2})^{n-m}\Bigl((2-\frac{3p}{2})^{m}-(1-\frac{p}{2})^{m}\Bigr)}\cdot\frac{p}{2}.
    \end{split}
    \end{equation}
    To require the bit error to be exponentially suppressed. We first need $(2-\frac{3p}{2})^{m}-(1-\frac{p}{2})^{m}>0$. This is trivially true because $2-\frac{3p}{2}>1-\frac{p}{2}$ whenever $p<1$. Then what is required is that $(1-\frac{p}{2})^{n-m}(2-\frac{3p}{2})^{m}$ is exponentially large. This can be achieved in the large $n$ limit as long as $(1-\frac{p}{2})^{1-\frac{m}{n}}(2-\frac{3p}{2})^{\frac{m}{n}}=1+\delta$ for any small constant $\delta>0$. Then the solution $p_{tol}$ to equation $(1-\frac{p}{2})^{1-\frac{m}{n}}(2-\frac{3p}{2})^{\frac{m}{n}}=1$ serves as the maximal tolerable noise strength because, in the large $n$ limit, when bit error is exponentially suppressed, the protocol could tolerate any phase error $e_p<\frac{1}{2}$ since $1-H(e_p)>0$ for any $e_p<\frac{1}{2}$. 

    Now we consider the phase error rate to complete the proof for key rate. For pure EPR states and EPR states affected only by $Z$-errors, the average phase error rate is
    \begin{equation}
    \begin{split}
        \overline{e}_p&=\Pr\{\text{$Z$-error }|\text{ no $X$ or $Y$ error}\}\\
        &=\frac{\Pr\{\text{$Z$-error },\text{ no $X$ or $Y$ error}\}}{\Pr\{\text{ no $X$ or $Y$ error}\}}\\
        &=\frac{p/4}{1-p/2}.
    \end{split}
    \end{equation}
    Note that after diagonal twirling, those surviving EPR pairs with $X$-errors would have $1/2$ phase error rate due to the complete dephasing channel. Therefore the overall phase error is:
    \begin{equation}\label{eq:ep}
    \begin{split}
        e_p&=\frac{\frac{p/4}{1-p/2}(1-\frac{3p}{4})^{m}(1-\frac{p}{2})^{n-m}+\frac{1}{2}\cdot2^{-m}(1-(1-\frac{p}{2})^{n})}{(1-\frac{3p}{4})^{m}(1-\frac{p}{2})^{n-m}+2^{-m}(1-(1-\frac{p}{2})^{n})}.
    \end{split}
    \end{equation}
    Since the bit error rate is exponentially suppressed, the information-reconciliation cost per retained pair vanishes in the asymptotic limit. Privacy amplification then gives the remaining factor $1-H(e_p)$, so the asymptotic key-rate expression is
    \begin{equation}
        R=\Bigl((1-\frac{3}{4}p)^{m}(1-\frac{p}{2})^{n-m}+2^{-m}(1-(1-\frac{p}{2})^{n})\Bigr)(1-H(e_p))(1-m/n).
    \end{equation}
\end{proof}

\section{Simulation setup for quantum key distribution and repeaters}\label{appendix:QKD}

For the distance-dependent simulations of quantum key distribution (QKD) and quantum repeaters, we adopt a simplified physical channel model motivated by optical fiber experiments. In practical optical implementations, sources typically emit weak coherent pulses containing vacuum, single-photon, and multiphoton components, which are measured using threshold or photon-number-resolving detectors. For simplicity, we assume either an ideal single-photon source or an ideal photon-number-resolving detector capable of filtering out all multiphoton components, an assumption that is becoming increasingly realistic with recent advances in single-photon sources and detectors~\cite{lu2023quantumdotsinglephotonsourcesquantum}. This reduces the analysis to the single-photon qubit subspace.

Over an optical fiber link of length $l$ (in $\mathrm{km}$) with attenuation coefficient $\alpha$, the channel transmittance decays exponentially as $\eta = 10^{-\alpha l/10}$. Protocols typically post-select on detection clicks, which occur either when a transmitted photon arrives (with probability $\eta$) or when a lost photon triggers a spurious detector click, known as a dark count, with probability $(1-\eta)Y_0$, where $Y_0$ is the dark-count rate per pulse. We model each dark-count event by replacing the lost photon with a maximally mixed qubit state, analogous to an erasure channel in which the loss event is masked by a random click. Furthermore, we assume that bit-flip and phase-flip errors occur independently, characterized by the intrinsic optical bit error rate $e_d$ for nonzero-photon components and a distance-dependent phase-flip probability $p(l)$ determined by an effective phase-variance coefficient $\beta$. Classical post-processing in QKD also introduces an error-correction efficiency factor $f$ that captures the overhead of bit error reconciliation.

Following Ref.~\cite{AQT}, we select parameter values from recent experimental demonstrations~\cite{Ma18PM,Zhu24,Liu2026}:
\begin{itemize}
    \item \textbf{Optical fiber loss:} $\alpha = 0.21\,\mathrm{dB/km}$, which sets the channel transmittance $\eta = 10^{-\alpha l/10}$ over a fiber of length $l$ kilometers.
    \item \textbf{Dark-count rate:} $Y_0 = 3 \times 10^{-8}$, defined as the probability of a detector click when no photon is present.
    \item \textbf{Intrinsic bit error rate:} $e_d$, the bit error rate for transmitted photons. We treat $e_d$ as the varied parameter in the QKD simulations and fix $e_d = 1.5\%$ in the quantum repeater simulations.
    \item \textbf{Effective phase-variance coefficient:} $\beta = 2.82 \times 10^{-3}\,\mathrm{km}^{-1}$, calibrated from experimental visibility data from Ref.~\cite{Liu2026}.
    \item \textbf{Error-correction efficiency:} $f = 1.06$~\cite{Zhu24}, representing the key-rate overhead of classical bit error correction.
\end{itemize} 

With these parameters established, we now define the effective single-photon bit error rate $e_b$ and phase error rate $e_p$ conditioned on a detection click. A transmitted photon arrives with relative weight $\eta$ and contributes an error rate $e_d$, whereas a dark-count event occurs with relative weight $(1-\eta)Y_0$ and yields a maximally mixed state with an error rate of $1/2$. Normalizing by the total click probability $\eta + (1-\eta)Y_0$, the single-photon bit error rate is
\begin{equation}
    e_b = \frac{\eta e_d + \frac{1}{2}(1-\eta)Y_0}{\eta + (1-\eta)Y_0}.
\end{equation}

For the phase error, we adopt the model from Ref.~\cite{Liu2026} and represent the residual differential phase after active stabilization as a zero-mean Gaussian random variable $\phi$ with variance $\sigma_\phi^2(l)$. Averaging the phase factor gives
\begin{equation}
    \mathbb{E}[e^{i\phi}]=e^{-\sigma_\phi^2(l)/2}.
\end{equation}
The resulting dephasing channel therefore has phase-flip probability
\begin{equation}
    p(l)=\frac{1-e^{-\sigma_\phi^2(l)/2}}{2},
\end{equation}
consistent with the phase-noise model in the Supplementary Information of Ref.~\cite{Liu2026}. To describe the distance dependence, we assume approximately independent residual phase contributions from successive fibre segments, so that their variances add:
\begin{equation}
    \sigma_\phi^2(l)=\beta l.
\end{equation}
We calibrate this effective coefficient to the interference visibility $V_0=0.986$ reported in Ref.~\cite{Liu2026} for a fibre link of length $l_0=10\,\mathrm{km}$. Attributing the visibility reduction to this effective phase noise, we use $V_0=e^{-\beta l_0/2}$ and obtain
\begin{equation}
    \beta=-\frac{2\ln V_0}{l_0}=-\frac{2\ln(0.986)}{10\,\mathrm{km}}\simeq2.82\times10^{-3}\,\mathrm{km}^{-1}.
\end{equation}
Consequently,
\begin{equation}
    p(l)=\frac{1-e^{-\beta l/2}}{2}.
\end{equation}
This additive-phase-variance model provides a reasonable phenomenological extrapolation of the calibrated phase noise to longer links for our numerical comparisons. Taking the dark count into account, the overall phase error rate is therefore,
\begin{equation}
    e_p=\frac{\eta(1-e^{-\beta l/2})/2+\frac{1}{2}(1-\eta)Y_0}{\eta+(1-\eta)Y_0}.
\end{equation}
For QKD, the maximum transmission distance is obtained when the net key rate
\begin{equation}
    R_{\textrm{one}}\propto 1-fH(e_b)-H(e_p)\leq 0.
\end{equation}
Here $H(e)$ is the Shannon entropy. We add an error correction efficiency factor $f$ to one of such entropy to represent the additional overhead by the error correction step. For two way protocols with recurrence method, we first compute the bit error rate $e_b'$ and phase error rate $e_p'$ after given rounds of recurrence. Then, similary, the maximum transmission distance is given by the point where 
\begin{equation}
    R_{\mathrm{two}}\propto 1-fH(e_b')-H(e_p')\leq 0.
\end{equation}
For Hamiltonian-based protocol, we choose a $15$-qubit diagonal Hamiltonian and measure $12$ qubits in the Hadamard basis. Since here we are working with finite Hamiltonian. Even though the bit error rate is still significantly suppressed, we cannot directly overlook it as done in the proof of Theorem~\ref{lemma:diagonal+classical}. We calculate the bit error rate and phase error rate $e_b^{H},e_p^{H}$ using Eq.~\eqref{eq:eb} and Eq.~\eqref{eq:ep}. Similarly the maximum transmission distance is determined when
\begin{equation}
    R_{\mathrm{H}}\propto 1-fH(e_b^{H})-H(e_p^{H})\leq 0.
\end{equation}
For repeaters, we directly calculate the bit and phase error rate and then report the fidelity, both without distillation and that after two-way distillation methods including recurrence method or Hamiltonian-based method.

\end{document}